\documentclass[11pt]{article}
\usepackage{amsmath,amsfonts,amsthm,amssymb, authblk,color,fullpage}
\usepackage{url}
\usepackage{graphicx,amssymb,amsmath,amsthm}
\usepackage{enumerate}
\usepackage{dsfont}
\usepackage{booktabs}
\usepackage{mathrsfs}
\usepackage{verbatim}
\usepackage{epsfig}
\usepackage{lscape}
\usepackage{subfigure}
\usepackage[colorlinks, linkcolor=red, anchorcolor=blue, citecolor=green]{hyperref}
\usepackage{epstopdf}
\usepackage{color}
\usepackage{caption}
\usepackage{algorithm}
\usepackage{algpseudocode}

\usepackage{graphicx}
\usepackage{subfigure}
\usepackage{cleveref}
\graphicspath{{Figures/}{logo/}}

\newtheorem{definition}{Definition}[section]

\newtheorem{prop}{Proposition}[section]
\newtheorem{theorem}{Theorem}[section]
\newtheorem{lemma}{Lemma}[section]

\newtheorem{remark}{Remark}[section]

\newtheorem{problem}{Problem}[section]

\date{}

\newcommand\blfootnote[1]{%
  \begingroup
  \renewcommand\thefootnote{}\footnote{#1}%
  \addtocounter{footnote}{-1}%
  \endgroup
}

\newcommand{\Addresses}{{
  \bigskip
  \footnotesize

  Jian-Feng Cai, \textsc{Department of Mathematics, The Hong Kong University of Science and Technology, Clear Water Bay, Kowloon, Hong Kong SAR, China}\par\nopagebreak
  \textit{E-mail address}, Jian-Feng Cai: \texttt{jfcai@ust.hk}

  \medskip

  Zhiqiang Xu , \textsc{LSEC, Inst.~Comp.~Math., Academy of
Mathematics and System Science,  Chinese Academy of Sciences, Beijing, 100091, China
\newline
School of Mathematical Sciences, University of Chinese Academy of Sciences, Beijing 100049, China}\par\nopagebreak
  \textit{E-mail address}, Zhiqiang Xu: \texttt{xuzq@lsec.cc.ac.cn}

  \medskip

  Zili Xu, \textsc{Department of Mathematics, The Hong Kong University of Science and Technology, Clear Water Bay, Kowloon, Hong Kong SAR, China}\par\nopagebreak
  \textit{E-mail address}, Zili Xu: \texttt{xuzili@ust.hk}
}}

\begin{document}
\baselineskip 14pt
\bibliographystyle{plain}

\title{Interlacing Polynomial Method for the Column Subset Selection Problem
\blfootnote{The work of Jian-Feng Cai was supported in part by the Hong Kong Research Grants Council GRF under Grants 16310620 and 16306821, and in part by Hong Kong Innovation and Technology Fund MHP/009/20. The work of Zhiqiang Xu is supported by the National Science Fund for Distinguished Young Scholars (12025108) and the National Nature Science Foundation of China  (12021001, 12288201).}
}

\author{Jian-Feng Cai, Zhiqiang Xu, and Zili Xu}

\maketitle

\begin{abstract}This paper investigates the spectral norm version of the column subset selection problem. Given a matrix $\mathbf{A}\in\mathbb{R}^{n\times d}$ and a positive integer $k\leq\text{rank}(\mathbf{A})$, the objective is to select exactly $k$ columns of $\mathbf{A}$ that minimize the spectral norm of the residual matrix after projecting $\mathbf{A}$ onto the space spanned by the selected columns. We use the method of interlacing polynomials introduced by Marcus-Spielman-Srivastava to derive a new upper bound on the minimal approximation error. This new bound is asymptotically sharp when the matrix $\mathbf{A}\in\mathbb{R}^{n\times d}$ obeys a spectral power-law decay. 
The relevant expected characteristic polynomials can be written as an extension of the expected polynomial for the restricted invertibility problem, incorporating two extra variable substitution operators.
Finally, we propose a deterministic polynomial-time algorithm that achieves this error bound up to a computational error.

\end{abstract}

\section{Introduction}

\subsection{The column subset selection problem}

The column subset selection has gained significant attention recently due to its wide applications in machine learning \cite{GE03,BMD09}, scientific computing  \cite{CH92,DMM08}, and signal processing \cite{BRN10}. This problem involves approximating a data matrix by selecting a limited number of columns from the matrix. Specifically, it can be formulated as follows:
\begin{problem}{\rm (The column subset selection problem)}\label{problem1.1}
Given a matrix $\mathbf{A}\in\mathbb{R}^{n\times d}$ of rank $t\leq \min\{n,d\}$ and a positive integer $k\leq t$, find a subset $S\subset[d]:=\{1,\ldots,d\}$ of size $k$ such that the following residual
\begin{equation*}
\Vert\mathbf{A}-\mathbf{A}_S\mathbf{A}_S^{\dagger}\mathbf{A} \Vert_{\xi}= \min\limits_{ \mathbf{X}\in\mathbb{R}^{k\times d}} \Vert\mathbf{A}-\mathbf{A}_S\mathbf{X} \Vert_{\xi}
\end{equation*}
is minimized over all possible $\binom{d}{k}$ choices for the subset $S$. Here, $\mathbf{A}_S\in \mathbb{R}^{n\times k}$ denotes the column submatrix of $\mathbf{A}$ consisting of columns indexed in the set $S$, $\mathbf{A}_S^{\dagger}\in \mathbb{R}^{k\times n}$ denotes the Moore-Penrose pseudoinverse of $\mathbf{A}_S$, and  $\xi=2$ or $\rm F$ denotes the spectral or Frobenius norm.
\end{problem}

Problem \ref{problem1.1} aims to construct the \emph{best} low-rank approximation $\mathbf{A}_S\mathbf{A}_S^{\dagger}\mathbf{A}$ of an input matrix $\mathbf{A}$ by projecting $\mathbf{A}$ onto the space spanned by a small number of carefully selected columns in $\mathbf{A}$. Note that the columns of $\mathbf{A}_S\mathbf{A}_S^{\dagger}\mathbf{A}$ are linear combinations of the columns of $\mathbf{A}_S$. Thus, compared with other low-rank approximation methods, such as truncated singular value decomposition, the column subset selection can take advantage of the sparsity of the input matrix and produce results that are easily interpretable in terms of the input matrix. Various algorithms have been proposed to address this problem  \cite{GE96,DRVW06,DR10,BMD09}.

\subsection{Our contributions}

%\subsubsection{The column subset selection problem}

In this paper, we investigate the column subset selection problem from a theoretical perspective, with particular emphasis on the spectral norm, i.e., $\xi=2$. Our primary result for Problem \ref{problem1.1} is presented in the following theorem.

\begin{theorem}\label{thm1.1}
Let $\mathbf{A}=[\mathbf{a}_1,\ldots,\mathbf{a}_d]\in\mathbb{R}^{n\times d}$ be a matrix of rank $t\leq \min\{d,n\}$. For each $1\leq i\leq t$, let $\lambda_i$ be the $i$-th largest eigenvalue value of $\mathbf{A}^{\rm T}\mathbf{A}$. Assume that $\lambda_t<\lambda_1$. Let $\alpha$ and $\beta$ be two real numbers such that
\begin{equation*}
\alpha^{-1}=\frac{1}{t} \sum_{i=1}^t \lambda_i^{-1}\quad\text{\rm and}\quad  \beta=\frac{\lambda_t^{-1}-\alpha^{-1}}{\lambda_t^{-1}-\lambda_1^{-1}}\in(0,1).
\end{equation*}
Then for any positive integer $k$ satisfying $ \beta\cdot t\leq k< t$, there exists a subset $S\subset[d]$ of size $k$ such that $\text{\rm rank}(\mathbf{A}_S)=k$ and
\begin{equation}\label{thm1.1:eq1}
\Vert\mathbf{A}-\mathbf{A}_{S}\mathbf{A}_{S}^{\dagger}\mathbf{A} \Vert_2^2\leq \frac{1}{1+(\Vert\mathbf{A}\Vert_2^2\cdot \alpha^{-1}-1)\cdot \gamma_{\mathbf{A},k}}\cdot \Vert\mathbf{A}\Vert_2^2,
\end{equation}
where $\gamma_{\mathbf{A},k}:=\bigg(\sqrt{\frac{k}{t}}-\sqrt{\frac{\beta}{1-\beta}\cdot (1-\frac{k}{t})} \bigg)^2\in[0,1)$.
\end{theorem}

 The right-hand side of \eqref{thm1.1:eq1} can be expressed as a weighted harmonic mean of $\Vert\mathbf{A}\Vert_2^2$ and $\alpha$, i.e., $1/(\frac{1-\gamma_{\mathbf{A},k}}{\Vert\mathbf{A}\Vert_2^2}+\frac{\gamma_{\mathbf{A},k}}{\alpha})$. It is easy to observe that $\gamma_{\mathbf{A},k}$ increases from $0$ to $1$ as $k$ increases from $\beta\cdot t$ to $t$. Consequently, the bound in \eqref{thm1.1:eq1} gradually decreases from $\Vert\mathbf{A}\Vert_2^2$ to $\alpha$ as $k$ increases from $\beta\cdot t$ to $t$. In Proposition \ref{prop3.1}, we will demonstrate that $\alpha$ is the expected value of $\Vert \mathbf{A}-\mathbf{A}_S\mathbf{A}_S^{\dagger}\mathbf{A} \Vert_2^2$ over the volume sampling distribution on the $(t-1)$-subsets of $[d]$.
This implies that $\alpha$ is a sharp upper bound on the minimum residual $\Vert \mathbf{A}-\mathbf{A}_S\mathbf{A}_S^{\dagger}\mathbf{A} \Vert_2^2$ for $|S|=t-1$, so the approximation error produced by Theorem \ref{thm1.1} is well upper bounded when $k$ is close to $t-1$.

%Additionally, it is worth noting that the bound stated in equation \eqref{thm1.1:eq1} is asymptotically sharp in two commonly encountered scenarios. 
To gain an understanding of the bound presented in Theorem 1.1, let's consider  the following two commonly encountered scenarios.
The first of these is a deterministically challenging problem that was constructed in \cite{BDM14}. The other scenario is when working with real-world data matrices that display a power-law decay of singular values.
\begin{enumerate}  
\item[(i)] Consider the matrix
\begin{equation}\label{thm1m1-remark:eq1}
\mathbf{A}=[\mathbf{e}_1+\delta\cdot\mathbf{e}_2,\ldots,\mathbf{e}_1+\delta\cdot\mathbf{e}_{d+1}]\in \mathbb{R}^{(d+1)\times d},
\end{equation}
where $\delta>0$ is a constant and $\mathbf{e}_1,\ldots,\mathbf{e}_{d+1}$ are the standard basis vectors in $ \mathbb{R}^{d+1}$.
This matrix is constructed in \cite{BDM14} as an example to demonstrate the error lower bound of the column subset selection problem.  A simple calculation shows that the rank of $\mathbf{A}$ is $d$ and $\mathbf{A}^{\rm T}\mathbf{A}=\delta^2\cdot \mathbf{I}_d+\mathbf{J}_d$, where $\mathbf{I}_d$ is the identity matrix and $\mathbf{J}_d$ is the all-ones square matrix of size $d$.  Furthermore, we have
\[
\Vert\mathbf{A}\Vert_2^2=\lambda_1=d+\delta^2, \lambda_2=\cdots=\lambda_d=\delta^2, \alpha=\frac{\delta^2(d+\delta^2)}{d-1+\delta^2},\,\, \beta=\frac{1}{d}.
\]
 Then, for any integer $1\leq k<d$ we can rewrite \eqref{thm1.1:eq1} as
\begin{equation}\label{thm1m1-remark:eq2}
\Vert\mathbf{A}-\mathbf{A}_{S}\mathbf{A}_{S}^{\dagger}\mathbf{A} \Vert_2^2\leq \frac{1}{1+\frac{1}{\delta^2}\cdot\bigg(\sqrt{k-\frac{k}{d}}-\sqrt{1-\frac{k}{d}} \bigg)^2}\cdot \Vert\mathbf{A}\Vert_2^2.
\end{equation}
On the other hand, as stated in \cite[Theorem 9.1]{BDM14}, for any subset $S\subset[d]$ of size $k<d$, we have
\begin{equation}\label{thm1m1-remark:eq3}
\Vert\mathbf{A}-\mathbf{A}_{S}\mathbf{A}_{S}^{\dagger}\mathbf{A} \Vert_2^2\geq \frac{\delta^2}{k+\delta^2}\cdot (d+\delta^2)=\frac{1}{1+\frac{1}{\delta^2}\cdot k}\cdot \Vert\mathbf{A}\Vert_2^2.
\end{equation}
Then we see that the upper bound in \eqref{thm1m1-remark:eq2} asymptotically matches the lower bound in \eqref{thm1m1-remark:eq3} when $d \rightarrow \infty$, suggesting that Theorem \ref{thm1.1} is asymptotically sharp for the matrix $\mathbf{A}$ defined in \eqref{thm1m1-remark:eq1}.

\item[(ii)] Consider the case where the matrix $\mathbf{A}\in\mathbb{R}^{n\times d}$ obeys a spectral power-law decay, that is, $\frac{c_1}{i^s}\leq \lambda_i\leq \frac{c_2}{i^{s}},i=1,\ldots,t$, where $c_1,c_2$ are two positive constants and $s>1$ is the decay rate exponent. The spectral power-law decay has been consistently observed across a wide range of real-world datasets of interest \cite{ADMMWG15,DKM20}. Assume that $1<\frac{c_2}{c_1}<1+\frac{1}{s+1}$. We use  $\sum_{i=1}^t i^{s}=\frac{t^{s+1}}{s+1}+o(t^{s+1})$ \cite{MP07} to obtain that $\frac{t^{s}}{c_2(s+1)}+o(t^{s})\leq \alpha^{-1}\leq \frac{t^{s}}{c_1(s+1)}+o(t^{s})$. Therefore, when $t\to\infty$, we have
\begin{equation*}
\beta=\frac{\lambda_t^{-1}-\alpha^{-1}}{\lambda_t^{-1}-\lambda_1^{-1}} \leq\frac{\frac{t^s}{c_1}-(\frac{t^s}{c_2(s+1)}+o(t^{s}))}{\frac{t^s}{c_2}-c_1}\leq\frac{c_2}{c_1}-\frac{1}{s+1}+\varepsilon=:\beta_1,
\end{equation*}
where $\varepsilon>0$ is a small enough constant. Let $k=l\cdot t$ where $l\in(\beta_1,1)$ is a constant. Then $\gamma_{\mathbf{A},k}$ can be regarded as a constant when $t\to\infty$, thus the bound given in \eqref{thm1.1:eq1} simplifies to $O(\frac{\alpha}{\gamma_{\mathbf{A},k}})$. It should be noted that the best rank $k$ approximation of $\mathbf{A}$ under the spectral norm is $\mathbf{A}_k=\sum_{i=1}^{k}\sqrt{\lambda_i} \mathbf{u}_i\mathbf{v}_i^{\rm T}$, where $\mathbf{u}_i$ and $\mathbf{v}_i$ are the left and right singular vectors corresponding to the singular value $\sqrt{\lambda_i}$. A simple calculation shows that when $t\to\infty$,
\begin{equation*}
\left(\frac{c_1}{c_2(s+1)l^s}+o(1)\right)\alpha\leq {\Vert\mathbf{A}-\mathbf{A}_k \Vert_2^2}={\lambda_{k+1}}\leq \left(\frac{c_2}{c_1(s+1)l^s}+o(1)\right)\alpha.
\end{equation*}
This indicates that, as $t\to\infty$, the upper bound $O(\frac{\alpha}{\gamma_{\mathbf{A},k}})$ asymptotically matches the lower bound $\Vert\mathbf{A}-\mathbf{A}_k \Vert_2^2$, up to a constant factor. Given that the spectrum of real-world data usually decays rapidly \cite{ADMMWG15,DKM20}, Theorem \ref{thm1.1} provides a promising approach to column subset selection and its various applications.

\end{enumerate}

The proof of Theorem \ref{thm1.1} is based on the method of interlacing polynomials, which was introduced by Marcus, Spielman and Srivastava. They used this method to provide a positive solution to the Kadison-Singer problem \cite{inter2}, a renowned problem in $C^*$ algebras posed by Richard Kadison and Isadore Singer in  \cite{KS59}.
Many combinatorial questions are known to be equivalent to the Kadison-Singer problem, including the Anderson Paving Conjecture \cite{And1,And2,And3}, Weaver's ${\rm{KS}}_r$ Conjecture \cite{Weaver}, and the Bourgain-Tzafriri Restricted Invertibility Conjecture \cite{BT89,CT}.

Inspired by the proof of Theorem \ref{thm1.1}, we propose a deterministic polynomial-time algorithm to solve Problem \ref{problem1.1}, as outlined in Algorithm \ref{alg1} in Section \ref{section3.4}.  The following theorem states that the output of Algorithm \ref{alg1} can attain the bound of Theorem \ref{thm1.1} up to a certain computational error.

\begin{theorem}\label{thm3.2}
 Let $\mathbf{A}$ be a matrix in $\mathbb{R}^{n\times d}$, and let $t$, $\alpha$, $\beta$, and $\gamma_{\mathbf{A},k}$ be defined as defined in Theorem \ref{thm1.1}. For each integer $k$ satisfying $\beta\cdot t\leq k<t$, Algorithm \ref{alg1} can output a subset $S=\{j_1,\ldots,j_k\}\subset [d]$ such that
\begin{equation}\label{alg:eq1}
\Vert\mathbf{A}-\mathbf{A}_{S}\mathbf{A}_{S}^{\dagger}\mathbf{A} \Vert_2^2\,\,\leq \,\, 2k\varepsilon+\frac{\Vert\mathbf{A}\Vert_2^2}{1+({\Vert\mathbf{A}\Vert_2^2}\cdot{\alpha}^{-1}-1)\cdot \gamma_{\mathbf{A},k}}.
\end{equation}
If $d<n$, then the  time complexity of Algorithm  \ref{alg1} is $O(kdn^2+kd^{w+1}\log(d\vee\varepsilon^{-1}))$. Otherwise,   if $d\geq n$, the  time complexity is $O(kdn^w\log(n\vee\varepsilon^{-1}))$. Here, $a\vee b:=\max\{a,b\}$ for two real numbers $a,b$, and $w\in  (2,2.373)$ is the matrix multiplication complexity exponent.
\end{theorem}

\subsection{Comparison with related works}

\subsubsection{The column subset selection problem}

Let $\mathbf{A}$ be a matrix in $\mathbb{R}^{n\times d}$ of rank $t\leq \min\{n,d\}$. Assume that $\mathbf{A}$ has the singular value decomposition $\mathbf{A}=\sum_{i=1}^{t}\sqrt{\lambda_i} \mathbf{u}_i\mathbf{v}_i^{\rm T}$, where $\lambda_1\geq\cdots\geq\lambda_t>0$ are the positive eigenvalues of $\mathbf{A}^{\rm T}\mathbf{A}$, $\mathbf{u}_i\in\mathbb{R}^n$ and $\mathbf{v}_i\in\mathbb{R}^d$ are the associated left and right singular vectors respectively. For the spectral norm version of Problem \ref{problem1.1}, prior work mainly focuses on finding a subset $S\subset[d]$ of size $k<t$ that satisfies the following multiplicative bound:
\begin{equation*}
\Vert \mathbf{A}-\mathbf{A}_S \mathbf{A}_S^{\dagger}\mathbf{A}\Vert_2^2\leq p(k,d)\cdot \Vert \mathbf{A}-\mathbf{A}_k\Vert_2^2=p(k,d)\cdot \lambda_{k+1},
\end{equation*}
where $p(k,d)$ is a function on $k$ and $d$, and $\mathbf{A}_k$ is the best rank-$k$ approximation of $\mathbf{A}$ under the spectral norm, i.e. $\mathbf{A}_k=\sum_{i=1}^{k}\sqrt{\lambda_i} \mathbf{u}_i\mathbf{v}_i^{\rm T}$.
 Early work on the column subset selection problem was largely centered around rank-revealing decompositions \cite{Gol65,GE96}.
 In practice, the algorithm based on rank-revealing QR (RRQR) decomposition \cite{GE96,HP92} provides an efficient deterministic algorithm with the multiplicative bound
\begin{equation*}
\Vert \mathbf{A}-\mathbf{A}_S \mathbf{A}_S^{\dagger}\mathbf{A}\Vert_2^2\leq (1+c^2k(d-k))\cdot \Vert \mathbf{A}-\mathbf{A}_k\Vert_2^2,
\end{equation*}
and it runs in $O(ndk\log_cd)$ time where  $c>1$ is a constant. A comprehensive summary of the algorithms based on RRQR decomposition can be found in Table 2 of \cite{BMD09}.

A class of randomized methods is also proposed for the column subset selection problem \cite{FKV04,BMD09,DRVW06,RV07,DR10,BBC20}. These methods select columns by sampling from certain distributions over the columns of the input matrix $\mathbf{A}$.
A two-stage algorithm is provided in \cite{BMD09}, which combines RRQR based algorithms and $k$-leverage score sampling. Specifically, the $i$-th column of $\mathbf{A}$ was chosen with probability proportional to its $k$-leverage score $l_i^{(k)}:=\sum_{j=1}^{k}v_{j,i}^2$ where $v_{j,i}$ is the $i$-th entry of the vector $\mathbf{v}_j\in\mathbb{R}^d$.
Their algorithm produces a subset $S\subset[d]$ of size $k$ with a probability of at least 0.7, ensuring that  the following multiplicative bound is satisfied:
\begin{equation*}
\Vert \mathbf{A}-\mathbf{A}_S \mathbf{A}_S^{\dagger}\mathbf{A}\Vert_2^2\leq O(k^{\frac{3}{2}}(d-k)^{\frac12}\log k)\Vert \mathbf{A}-\mathbf{A}_k\Vert_2^2.
\end{equation*}
Deshpande, Rademacher, Vempala, and Wang \cite{DRVW06} gave an alternative approach to the column subset selection by using the volume sampling, i.e., picking a subset $S\subset[d]$ with probability proportional to $\det[\mathbf{A}_S^{\rm T} \mathbf{A}_S ]$. As a result, they proved that
\begin{equation*}
\mathbb{E}_S\ \Vert \mathbf{A}-\mathbf{A}_S \mathbf{A}_S^{\dagger}\mathbf{A}\Vert_2^2\leq (d-k)(k+1)\cdot \Vert \mathbf{A}-\mathbf{A}_k\Vert_2^2,
\end{equation*}
where $S$ is selected under the volume sampling on all $k$-subsets of $[d]$.
A polynomial time algorithm for the volume sampling was later provided by Deshpande and Rademacher  \cite{DR10}. Recently, Belhadji, Bardenet, and Chainais \cite{BBC20} proposed sampling from a projection determinantal point process and obtained an improved multiplicative bound:
\begin{equation*}
\mathbb{E}_S\ \Vert \mathbf{A}-\mathbf{A}_S \mathbf{A}_S^{\dagger}\mathbf{A}\Vert_2^2\leq (1+k(\tilde{d}-k))\cdot \Vert \mathbf{A}-\mathbf{A}_k\Vert_2^2,
\end{equation*}
where $\tilde{d}:=\# \{i\in[d]: l_i^{(k)}>0 \}$ is the number of the nonzero $k$-leverage scores.
For a comprehensive review of existing randomized sampling methods for the column subset selection problem, please refer to \cite{BBC20}.

Our upper bound in \eqref{thm1.1:eq1} offers at least two advantages over the existing multiplicative bounds from a theoretical point of view. First, our bound is strictly less than $ \Vert \mathbf{A}\Vert_2^2$ for $k>\beta\cdot \text{\rm rank}(\mathbf{A})$, while existing multiplicative bounds can be greater than $\Vert \mathbf{A}\Vert_2^2$ if $\lambda_{k+1}$ is sufficiently large.
Second, our bound in \eqref{thm1.1:eq1} is often much tighter than existing multiplicative bounds.
 For example, consider $\mathbf{A}$ defined as in \eqref{thm1m1-remark:eq1}, and let $k,d$ tend to infinity with the ratio $k/d$ fixed.
 Recall that in this case we have $\Vert \mathbf{A}-\mathbf{A}_k\Vert_2^2=\delta^2$. The multiplicative bounds mentioned above become $\Vert \mathbf{A}-\mathbf{A}_S \mathbf{A}_S^{\dagger}\mathbf{A}\Vert_2^2\leq O(k(d-k))\cdot \delta^2$, while, a simple calculation shows that  our bound in \eqref{thm1m1-remark:eq2} asymptotically  converges to  the lower bound $\Vert \mathbf{A}-\mathbf{A}_S \mathbf{A}_S^{\dagger}\mathbf{A}\Vert_2^2\geq \frac{d+\delta^2}{k+\delta^2}\cdot \delta^2=O(1)\cdot\delta^2$ presented in \eqref{thm1m1-remark:eq3}.

It is worth noting that in certain situations, existing multiplicative bounds may outperform our bound. Specifically, when both the dimensions $d$ and $t$ tend to infinity, and if the smallest nonzero eigenvalue $\lambda_t$ approaches zero at a faster rate compared to any other nonzero eigenvalue $\lambda_i$, then the value of $\beta$ tends to 1. Consequently, in this scenario, our bound \eqref{thm1.1:eq1} is not applicable for the majority of $k<\text{\rm rank}(\mathbf{A})$.

A large number of papers have been devoted to selecting a subset $S\subset[d]$ such that the submatrix $\mathbf{A}_S$ satisfies one of the following relative-error bounds:
\begin{align*}
&\Vert \mathbf{A}-\mathbf{A}_S \mathbf{A}_S^{\dagger}\mathbf{A}\Vert_{\xi}\leq (1+\varepsilon)\cdot \Vert \mathbf{A}-\mathbf{A}_k\Vert_{\xi},\\
&\Vert \mathbf{A}-\mathbf{A}_S \mathbf{A}_S^{\dagger}\mathbf{A}\Vert_{\xi}\leq \Vert \mathbf{A}-\mathbf{A}_k\Vert_{\xi}+\varepsilon \Vert \mathbf{A} \Vert_{\xi},
\end{align*}
where $\varepsilon>0$ is a given error parameter and $\xi\in\{\rm F,2\}$ \cite{DRVW06,RV07,DMM08,BDM14,WS18}. To achieve the relative error bound, the size of $S$ often needs to be larger than $k$ and  it is dependent on the error parameter $\varepsilon$.
In this paper, we primarily focus on deriving an upper bound for the residual $\Vert \mathbf{A}-\mathbf{A}_S \mathbf{A}_S^{\dagger}\mathbf{A}\Vert_2$ when the size of $S$ is specified. Therefore, the discussion of the relative-error bounds is beyond the scope of this paper.

\subsubsection{Approximating Matrices by Column Selection}

Friedland and Youssef investigated the problem of approximating matrices via a variant of column selection \cite{You3}. We recall that the stable rank of a matrix $\mathbf{B}$ is defined by $\text{\rm srank}(\mathbf{B}):=\Vert \mathbf{B}\Vert_{\rm F}^2/\Vert \mathbf{B}\Vert_2^2$. Based on Marcus-Spielman-Srivastava's solution to the Kadison-Singer problem, Friedland and Youssef \cite[Theorem 1.1]{You3} prove that for any $\varepsilon>0$ and any matrix $\mathbf{A}\in\mathbb{R}^{n\times d}$, there exists a nonnegative diagonal matrix $\mathbf{D}\in\mathbb{R}^{n\times n}$ with at most $\text{\rm srank}(\mathbf{A})/c\varepsilon^2$ nonzero entries such that
\begin{equation*}
\Vert  \mathbf{A}^{\rm T} \mathbf{A}-\mathbf{A}^{\rm T}\mathbf{D}\mathbf{A}  \Vert_2 	\leq \varepsilon \Vert \mathbf{A}\Vert_2^2.
\end{equation*}
This implies that the spectrum of $\mathbf{D}^{\frac12}\mathbf{A}\in\mathbb{R}^{n\times d}$ is close to that of $\mathbf{A}$ (see also \cite{inter0}).

The similarity between Theorem \ref{thm1.1} and \cite[Theorem 1.1]{You3} is that they both produce a sketch matrix $\widetilde{\mathbf{A}}$ of an input matrix $\mathbf{A}\in\mathbb{R}^{n\times d}$ such that the approximation error $\Vert  \mathbf{A}^{\rm T} \mathbf{A}-\widetilde{\mathbf{A}}^{\rm T}\widetilde{\mathbf{A}}\Vert_2$ is well upper bounded by a constant times $\Vert\mathbf{A}\Vert_2^2$, i.e.,
\begin{equation}\label{section1.4.2:eq1}
\Vert  \mathbf{A}^{\rm T} \mathbf{A}-\widetilde{\mathbf{A}}^{\rm T}\widetilde{\mathbf{A}}  \Vert_2\leq\varepsilon_{\mathbf{A},k}\cdot  \Vert\mathbf{A}\Vert_2^2,
\end{equation}
where $\varepsilon_{\mathbf{A},k}<1$ is a constant that depends on the matrix $\mathbf{A}$ and the sampling size $k$. Specifically, Theorem \ref{thm1.1} produces a sketch matrix $\widetilde{\mathbf{A}}=\mathbf{A}_S\mathbf{A}_S^{\dagger}\mathbf{A}$ by projecting $\mathbf{A}$ onto the space spanned by a small number of the columns in $\mathbf{A}$. Note that
\begin{equation*}
\Vert \mathbf{A}-\mathbf{A}_S\mathbf{A}_S^{\dagger}\mathbf{A}\Vert_2^2=\Vert  \mathbf{A}^{\rm T} \mathbf{A}- (\mathbf{A}_S\mathbf{A}_S^{\dagger}\mathbf{A})^{\rm T}(\mathbf{A}_S\mathbf{A}_S^{\dagger}\mathbf{A}) \Vert_2,
\end{equation*}
so the sketch matrix $\widetilde{\mathbf{A}}=\mathbf{A}_S\mathbf{A}_S^{\dagger}\mathbf{A}$ produced by Theorem \ref{thm1.1} satisfies the bound \eqref{section1.4.2:eq1} with $\varepsilon_{\mathbf{A},k}=1/(1+(\Vert\mathbf{A}\Vert_2^2\cdot{\alpha}^{-1}-1)\cdot\gamma_{\mathbf{A},k})$. On the other hand, \cite[Theorem 1.1]{You3} produces a sketch matrix $\widetilde{\mathbf{A}}=\mathbf{D}^{\frac12}\mathbf{A}$ by selecting a small number of rescaled rows in $\mathbf{A}$, and this sketch matrix satisfies the bound \eqref{section1.4.2:eq1} with $\varepsilon_{\mathbf{A},k}=O(\sqrt{{\rm srank}(\mathbf{A})/k})$.
Since the sampling methods used in Theorem \ref{thm1.1} and \cite[Theorem 1.1]{You3} are distinct, the bound in \eqref{thm1.1:eq1} cannot be directly compared to the bound in Friedland-Youssef's work.

\subsubsection{The Restricted Invertibility Principle}

Many variants of column subset selection have been studied in numerical linear algebra and computer science. One of these results is the Bourgain-Tzafriri restricted invertibility principle \cite{BT89}, which is related to the following subset selection problem.

\begin{problem}{\rm (The restricted invertibility problem)}\label{problem1.3}
Given a matrix $\mathbf{B}\in\mathbb{R}^{n\times d}$ and a positive integer $l\leq \text{\rm rank}(\mathbf{B})$, find a subset $R\subset[d]$ of size $l$ such that the smallest singular value of $\mathbf{B}_R$, denoted by $\sigma_{\rm \min }(\mathbf{B}_R)$, is maximized over all possible $\binom{d}{l}$ choices for the subset $R$.
\end{problem}

Roughly speaking, Problem \ref{problem1.3} aims to select $l$ columns that are as ``linearly independent'' as possible. The Bourgain-Tzafriri restricted invertibility principle shows that there exist universal constants $c_1, c_2\in(0,1)$ such that for any invertible matrix $\mathbf{B}\in\mathbb{R}^{d\times d}$ with normalized columns, one can find a subset $R\subset[d]$ of size at least $c_1\cdot \text{\rm srank}(\mathbf{B})$ such that $\sigma_{\rm \min }(\mathbf{B}_R)\geq c_2$. This result was later strengthened by the works of \cite{Ver01,You14,SS12,You2,inter3,ravi1}.

In Section \ref{section4.1}, we will demonstrate a direct link between Problem \ref{problem1.3} and a specific instance of Problem \ref{problem1.1}.
More specifically, we show that if $\mathbf{A}\in \mathbb{R}^{n\times d}$ is a rank $d$ matrix with the singular value decomposition $\mathbf{A}=\mathbf{U}\mathbf{\Sigma}\mathbf{V}^{\rm T}$, then for each subset $S\subset[d]$ we have
 \begin{equation*}
 \Vert \mathbf{A}-\mathbf{A}_S\mathbf{A}_S^{\dagger}\mathbf{A} \Vert_2^2=\frac{1}{\sigma_{\rm min}^2(\mathbf{B}_{S^C})},	
 \end{equation*}
 where $\mathbf{B}=\mathbf{U}\mathbf{\Sigma}^{-1}\mathbf{V}^{\rm T}\in\mathbb{R}^{n\times d}$. This implies that minimizing the residual $\Vert \mathbf{A}-\mathbf{A}_S\mathbf{A}_S^{\dagger}\mathbf{A} \Vert_2^2 $ is equivalent to maximizing $\sigma_{\rm min}^2(\mathbf{B}_{S^C})$ where the subset $S$ ranges over all $k$-subsets of $[d]$.
 Consequently, in the case of $\text{\rm rank}(\mathbf{A})=d$, Problem \ref{problem1.1} can be reformulated as Problem \ref{problem1.3} with $\mathbf{B}=\mathbf{U}\mathbf{\Sigma}^{-1}\mathbf{V}^{\rm T}$ and $l=d-k$.

\subsection{Technical overview}

Let $\mathbf{A}$ be a matrix in $\mathbb{R}^{n\times d}$. For each subset $S\subset[d]$, we set $\mathbf{Q}_S:=\mathbf{I}_n-\mathbf{A}_S\mathbf{A}_S^{\dagger}\in \mathbb{R}^{n\times n}$ which is the projection matrix of the orthogonal complement of the column space of the matrix $\mathbf{A}_S$. For each subset $S\subset[d]$ we define the degree $d$ polynomial
\begin{equation}\label{section1.3:eq1}
p_S(x):=\det[x\cdot\mathbf{I}_d-(\mathbf{A}-\mathbf{A}_{S}\mathbf{A}_{S}^{\dagger}\mathbf{A})^{\rm T}(\mathbf{A}-\mathbf{A}_{S}\mathbf{A}_{S}^{\dagger}\mathbf{A})]=\det[x\cdot\mathbf{I}_d-\mathbf{A}^{\rm T}\mathbf{Q}_{S}\mathbf{A}].
\end{equation}
It immediately follows that
\begin{equation}\label{section3.1:eq2}
\text{\rm maxroot}\ p_S(x)=\Vert \mathbf{A}-\mathbf{A}_{S}\mathbf{A}_{S}^{\dagger}\mathbf{A}\Vert_2^2.
\end{equation}
We employ the method of interlacing polynomials to prove the following theorem.

\begin{theorem}\label{thm1.3}
Let $\mathbf{A}$ be a matrix in $\mathbb{R}^{n\times d}$. Then for each positive integer $k\leq\text{\rm rank}(\mathbf{A})$, there exists a subset $S\subset[d]$ of size $k$ such that $\text{\rm rank}(\mathbf{A}_S)=k$ and
\begin{equation*}
\text{\rm maxroot}\ p_{S}(x)\leq \text{\rm maxroot}\ (x^2\cdot\partial_x-d\cdot x)^k \ \det[x\cdot \mathbf{I}_d-\mathbf{A}^{\rm T}\mathbf{A}].
\end{equation*}	

\end{theorem}

We will prove that the polynomial $(x^2\cdot\partial_x-d\cdot x)^k \ \det[x\cdot \mathbf{I}_d-\mathbf{A}^{\rm T}\mathbf{A}]$ in Theorem \ref{thm1.3} has two important properties. First, when $k\leq\text{\rm rank}(\mathbf{A})$, this polynomial is a degree $d$ polynomial with only nonnegative real roots, and its largest root is non-increasing as $k$ increases (see Section \ref{section2.3}). Second, this polynomial is the expectation of the polynomials $p_S(x)$ over the volume sampling of $k$-subsets of $[d]$, up to a constant factor, i.e., (see Proposition \ref{prop3.1})
\begin{equation}\label{expectation-ps}
k!\sum\limits_{S\subset[d],\vert S\vert=k}^{}\det[\mathbf{A}_S^{\rm T}\mathbf{A}_S]\cdot p_S(x)=(x^2\cdot\partial_x-d\cdot x)^k\  \det[x\cdot\mathbf{I}_d-\mathbf{A}^{\rm T}\mathbf{A}].
\end{equation}
Therefore, Theorem \ref{thm1.3} basically says that the largest root of the expected polynomial $(x^2\cdot\partial_x-d\cdot x)^k \ \det[x\cdot \mathbf{I}_d-\mathbf{A}^{\rm T}\mathbf{A}]$ provides an upper bound on the minimal largest root of $p_S(x)$, with $S$ ranging over all $k$-subsets of $[d]$.

To finish the proof of Theorem \ref{thm1.1}, we employ  the univariate barrier method introduced in \cite{inter0,SS12,inter3,inter5} to estimate the largest root of the expected polynomial $(x^2\cdot\partial_x-d\cdot x)^k\  \det[x\cdot\mathbf{I}_d-\mathbf{A}^{\rm T}\mathbf{A}]$.

\subsection{Organization}
The rest of this paper is organized as follows.
In Section \ref{section2}, we present several fundamental linear algebraic and analytic facts. In Section \ref{section3}, we prove Theorems \ref{thm1.1}, \ref{thm3.2} and \ref{thm1.3}.

\section{Preliminaries}\label{section2}

\subsection{Notations }

We first introduce some notations. For a positive integer $d$, we use $[d]$ to denote the set $\{1,2,\ldots,d\}$. For any two integers $k$ and $d$ satisfying that $k<d$, we use $[k,d]$ to denote the set $\{k,k+1,\ldots,d\}$.  
For a subset $T\subset[d]$ we use $|T|$ to denote the size of $T$, and we say $T$ is a $k$-subset of $[d]$ if $|T|=k$. 

For a positive integer $d$, we use $\mathbf{I}_d$ to denote the identity matrix of size $d$. For a matrix $\mathbf{A}\in\mathbb{R}^{d\times d}$, we denote  the trace of $\mathbf{A}$ by $\text{\rm Tr}[\mathbf{A}]$, i.e., $\text{\rm Tr}[\mathbf{A}]=\sum_{i=1}^d \mathbf{A}(i,i)$. We denote the Euclidean norm of a vector $\mathbf{x}$ by $\Vert\mathbf{x}\Vert$. For a matrix $\mathbf{A}\in\mathbb{R}^{n\times d}$ and a $k$-subset $S\subset[d]$, we use $\mathbf{A}_S\in\mathbb{R}^{n\times k}$ to denote the submatrix of $\mathbf{A}$ obtained by extracting the columns of $\mathbf{A}$ indexed by $S$. For a square matrix $\mathbf{B}\in\mathbb{R}^{d\times d}$ and a $k$-subset $S\subset[d]$, we use $\mathbf{B}(S)\in\mathbb{R}^{k\times k}$ to denote the submatrix of $\mathbf{B}$ obtained by extracting the rows and columns of $\mathbf{B}$ indexed by $S$. For a matrix $\mathbf{A}$, we use $\sigma_{\min}(\mathbf{A})$ and $\sigma_{\max}(\mathbf{A})$ to denote the smallest and largest singular value of $\mathbf{A}$, respectively. If $\mathbf{A}$ is a square matrix, then we use $\lambda_{\min}(\mathbf{A})$ and $\lambda_{\max}(\mathbf{A})$ to denote the smallest and largest eigenvalue of $\mathbf{A}$, respectively. We define $\mathrm{Ran}(\mathbf{A})$ to be the range, or the column space, of a given matrix $\mathbf{A}$.

We use $\mathbb{P}^+(d)$ to denote  the family of real-rooted univariate polynomials of degree exactly $d$ that have a positive leading coefficient and have only nonnegative roots. We use $\mathbb{R}[x]$ to denote  the family of univariate polynomials with real coefficients. Let $z_1,\ldots,z_d$ be variables. We use $\mathbb{R}[z_1,\ldots,z_d]$ to denote the family of multivariate polynomials with only real coefficients.
 We write $\partial_{z_i}$ to indicate the partial differential $\partial\slash\partial_{z_i}$. For each subset $S\subset[d]$, we use $z^S$ to denote the multi-affine polynomial $\prod_{i\in S}z_i$.

Let $p(x)=\sum_{i=0}^{d}a_i\cdot x^i\in\mathbb{R}[x]$ be a univariate polynomial of degree at most $d$, where $a_0,\ldots,a_d$ are constants. We define the flip operator $\mathcal{R}_{x,d}^+:\mathbb{R}[x]\to \mathbb{R}[x]$ as
\begin{equation}\label{notation:eq2}
\mathcal{R}_{x,d}^+\ p(x) :=x^d\cdot p(1/x)=\sum_{i=0}^{d}a_i\cdot x^{d-i}.
\end{equation}
Throughout this paper we always assume that the polynomial $\mathcal{R}_{x,d}^+\ p(x)$ is well-defined at $x=0$, that is, if $p(x)=\sum_{i=0}^{d}a_i\cdot x^i$ then $(\mathcal{R}_{x,d}^+\ p)(0)=a_d$.
We next generalize the definition of $\mathcal{R}_{x,d}^+$ to the multivariate case. Let $r_1,\ldots,r_d$ be a collection of nonnegative integers, and let $f(z_1,\ldots,z_d)\in\mathbb{R}[z_1,\ldots,z_d]$ be a multivariate polynomial of degree at most $r_j$ in $z_j$, $j=1,\dots,d$. For each $i\in[d]$ we define the flip operator $\mathcal{R}_{z_i,r_i}^+:\mathbb{R}[z_1,\ldots, z_d]\to \mathbb{R}[z_1,\ldots, z_d]$ as
\begin{equation}\label{notation:eq1}
\mathcal{R}_{z_i,r_i}^+\ f :=z_i^{r_i}\cdot f(z_1,\ldots,z_i^{-1},\ldots,z_d).
\end{equation}
Note that $\mathcal{R}_{z_i,r_i}^+\ f$ is a multivariate polynomial of degree at most $r_j$ in $z_j$, $j=1,\dots,d$.

\subsection{The Moore-Penrose pseudoinverse}

For a matrix $\mathbf{A}\in\mathbb{R}^{n\times d}$, we use $\mathbf{A}^{\dagger}\in \mathbb{R}^{d\times n}$ to denote the Moore-Penrose pseudoinverse of $\mathbf{A}$. Assume that $\text{\rm rank}(\mathbf{A})=t$ and $\mathbf{A}$ has the thin singular value decomposition: $\mathbf{A}=\mathbf{U}\mathbf{\Sigma}\mathbf{V}^{\rm T}$, where $\mathbf{\Sigma}\in\mathbb{R}^{t\times t}$, $\mathbf{U}\in\mathbb{R}^{n\times t}$, and $\mathbf{V}\in\mathbb{R}^{d\times t}$ satisfying that $\mathbf{U}^{\rm T}\mathbf{U}=\mathbf{V}^{\rm T}\mathbf{V}=\mathbf{I}_t$. Then we have $\mathbf{A}^{\dagger}=\mathbf{V}\mathbf{\Sigma}^{-1}\mathbf{U}^{\rm T}$, and the matrix $\mathbf{A}\mathbf{A}^{\dagger}=\mathbf{U}\mathbf{U}^{\rm T}$ is the projection matrix of the column space $\mathrm{Ran}(\mathbf{A})$. For a subset $S\subset[d]$, we use $\mathbf{A}_S^{\dagger}$ to denote the Moore-Penrose pseudoinverse of $\mathbf{A}_S$, i.e., $\mathbf{A}_S^{\dagger}:=(\mathbf{A}_S)^{\dagger}$.

The following lemmas are related to the Moore-Penrose pseudoinverse and will be utilized later. The first lemma presents the rank-one update formula for a projection matrix.

\begin{lemma}\label{lemma2.1}
Assume that $n\geq d>1$. Let $\mathbf{B}\in\mathbb{R}^{n\times d}$ and $\mathbf{b}\in\mathbb{R}^{n}$ satisfying
$\mathbf{b}\notin \text{\rm Ran}(\mathbf{B})$. Set $\mathbf{M}:=[\mathbf{B},\mathbf{b}]\in\mathbb{R}^{n\times (d+1)}$. Then we have
\begin{equation*}
\mathbf{M}\mathbf{M}^{\dagger}=\mathbf{B}\mathbf{B}^{\dagger}+\frac{(\mathbf{Q}\mathbf{b})(\mathbf{Q}\mathbf{b})^{\rm T}}{\Vert \mathbf{Q}\mathbf{b} \Vert^2},
\end{equation*}
where $\mathbf{Q}=\mathbf{I}_n-\mathbf{B}\mathbf{B}^{\dagger}$ is the projection matrix of the orthogonal complement of $\mathrm{Ran}(\mathbf{B})$.

\end{lemma}

\begin{proof}

Note that $\mathbf{M}\mathbf{M}^{\dagger}$ and $\mathbf{B}\mathbf{B}^{\dagger}$ are the projection matrices of $\mathrm{Ran}(\mathbf{M})$ and $\mathrm{Ran}(\mathbf{B})$, respectively. Then the conclusion follows from \cite[Theorem 4.5]{YTT11}.

\end{proof}

The following lemma is a restatement of \cite[Lemma 12]{DR10}, showing that a determinant can be divided into a product of two determinants.

\begin{lemma}{\rm \cite[Lemma 12]{DR10}}\label{lemma2.2}
Let $\mathbf{B}\in\mathbb{R}^{n\times d}$ and $\mathbf{C}\in\mathbb{R}^{n\times m}$.
Set $\mathbf{M}=[\mathbf{B},\mathbf{C}]\in\mathbb{R}^{n\times (d+m)}$ and $\mathbf{Q}=\mathbf{I}_n-\mathbf{B}\mathbf{B}^{\dagger}$.
Then, we have
\begin{equation*}
\det[\mathbf{M}^{\rm T} \mathbf{M}]=\det[\mathbf{C}^{\rm T}\mathbf{Q} \mathbf{C}]\cdot \det[\mathbf{B}^{\rm T} \mathbf{B}].
\end{equation*}	
\end{lemma}

\subsection{The differential operator $(x^2\cdot\partial_x-d\cdot x)$}\label{section2.3}

In this subsection, we give a brief introduction to the differential operator $(x^2\cdot\partial_x-d\cdot x)$, which is raised in Theorem \ref{thm1.3}. This differential operator can be related to the polar derivative at 0 of a degree-$d$ polynomial, which is defined by $(d-x\cdot\partial_x)\cdot p(x)$ for a degree-$d$ polynomial $p(x)$ \cite[p. 44]{Mar66}. Additionally, the differential operator $(x^2\cdot\partial_x-d\cdot x)$ is closely linked to the degree reduction  formula for the symmetric multiplicative convolution of polynomials, as introduced in \cite[Lemma 4.9]{inter5}.

The following proposition presents some basic properties about the differential operator $(x^2\cdot\partial_x-d\cdot x)$. We recall that a univariate polynomial $p(x)$ is considered to be {\em real-rooted} if all of its coefficients and roots are real.

\begin{prop}\label{prop21}
\begin{enumerate} 
\item[(i)] Assume that $p(x)$ is a polynomial with degree at most $d$. Then for any nonnegative integer $k$, we have
\begin{equation}\label{prop2.1:eq1}
(x^2\cdot\partial_x-d\cdot x)^k\ p(x)= (-1)^k\cdot \mathcal{R}_{x,d}^+\cdot \partial_x^k\cdot \mathcal{R}_{x,d}^+\ p(x).
\end{equation}

\item[(ii)]  Assume that $p(x)$ is a degree-$d$ real-rooted polynomial with exactly $t$ nonzero roots. Then $(x^2\cdot\partial_x-d\cdot x)^k\ p(x) \equiv 0 $ if and only if $k\geq t+1$. When $k=t$, $(x^2\cdot\partial_x-d\cdot x)^k\ p(x)$ is a multiple of $x^d$.
\item[(iii)]
 Assume that $p(x)$ is a polynomial in $\mathbb{P}^+(d)$ with exactly $t\geq 1$ positive roots. Then we have
 \begin{equation}\label{prop2.1:eq2}
\text{\rm maxroot}\	(x^2\cdot\partial_x-d\cdot x)^k\ p(x)= \frac{1}{\text{\rm minroot}\ \partial_x^k\cdot \mathcal{R}_{x,d}^+\ p(x)},\quad \text{ for  }\ k<t.
\end{equation}
 Furthermore, for any non-negative integer $k$ such that $k\leq t$,  the polynomial $(x^2\cdot\partial_x-d\cdot x)^k\ p(x)$ belongs to $\mathbb{P}^+(d)$ and its largest root is non-increasing as $k$ increases.
\end{enumerate}
\end{prop}
\begin{proof}

(i) An elementary calculation shows that
\begin{subequations}
\begin{align}
(x^2\cdot\partial_x-d\cdot x)\ p(x)&= - \mathcal{R}_{x,d}^+\cdot \partial_x\cdot \mathcal{R}_{x,d}^+\ p(x), \label{prop2.1:eq3}\\
\mathcal{R}_{x,d}^+\cdot \mathcal{R}_{x,d}^+\ p(x)&=p(x). \label{prop2.1:eq4}
\end{align}
\end{subequations}
Repeatedly using  \eqref{prop2.1:eq3} and \eqref{prop2.1:eq4}, we obtain \eqref{prop2.1:eq1}.

(ii) We write $p(x)=x^{d-t}\cdot q(x)$, where $q(x)$ is a degree $t$ polynomial with $t$ nonzero real roots. According to \eqref{prop2.1:eq1},  $(x^2\cdot\partial_x-d\cdot x)^k\ p(x)\equiv 0 $ if and only if $\partial_x^k\cdot \mathcal{R}_{x,d}^+\ p(x) \equiv 0 $. Noting that $\mathcal{R}_{x,d}^+\ p(x)=x^d\cdot p(1/x)=x^{t}\cdot q(1/x)$ has degree $t$, we obtain that
$ \partial_x^k\cdot \mathcal{R}_{x,d}^+\ p(x)\equiv 0 $ if and only if $k\geq t+1$. When $k=t$, note that $ \partial_x^t\cdot \mathcal{R}_{x,d}^+\ p(x)$ is a constant, so \eqref{prop2.1:eq1} implies that $(x^2\cdot\partial_x-d\cdot x)^k\ p(x)$ is a multiple of $x^d$.

(iii) We write $p(x)=\lambda_0\cdot x^{d-t}\cdot \prod_{i=1}^t (x-\lambda_i)$, where $\lambda_0,\lambda_1,\ldots,\lambda_t> 0$ are positive real numbers. Then $\mathcal{R}_{x,d}^+\ p(x)=x^d\cdot p(1/x)=(-1)^t\cdot \prod_{i=0}^t \lambda_i \cdot  \prod_{i=1}^t(x-\lambda_i^{-1})$ is a real-rooted polynomial of degree $t$ with only positive roots. According to Rolle's theorem, the polynomial $\partial_x^k\cdot \mathcal{R}_{x,d}^+\ p(x)$ has only positive roots when $k<t$. Then, \eqref{prop2.1:eq2} immediately follows from \eqref{prop2.1:eq1}.

We next assume $k\leq t$. According to Rolle's theorem, the smallest root of $\partial_x^k\cdot \mathcal{R}_{x,d}^+\ p(x)$ is non-decreasing in $k$. Recall that $\partial_x^k\cdot \mathcal{R}_{x,d}^+\ p(x)$ has only positive roots when $k<t$ and that $\partial_x^k\cdot \mathcal{R}_{x,d}^+\ p(x)$ is a constant when $k=t$. Then, \eqref{prop2.1:eq1} implies that $(x^2\cdot\partial_x-d\cdot x)^k\ p(x)$ has only nonnegative roots, and \eqref{prop2.1:eq2} shows that the largest root of  $(x^2\cdot\partial_x-d\cdot x)^k \ p(x)$ is non-increasing in $k$. Finally, by \eqref{prop2.1:eq1} we can calculate that $(x^2\cdot\partial_x-d\cdot x)^k\ p(x)$ has the positive leading coefficient $k!\cdot \lambda_0\cdot \sum_{S\subset[t],\vert S\vert=k} \prod_{i\in S}\lambda_i$ for each $k\leq t$. Hence, we can conclude that $(x^2\cdot\partial_x-d\cdot x)^k\ p(x)$ is in $\mathbb{P}^+(d)$ for each $k\leq t$, and its largest root is non-increasing in $k$.\end{proof}
\begin{remark}\label{remark2.1}
 The identity \eqref{prop2.1:eq1} can be generalized to the multivariate case. Let $r_1,\ldots,r_d$ be a collection of nonnegative integers, and let $f(z_1,\ldots,z_d)\in\mathbb{R}[z_1,\ldots,z_d]$ be a multivariate polynomial of degree at most $r_i$ in $z_i$, $i=1,\ldots,d$. Then a similar calculation shows that, for each $i\in[d]$ and for each integer $k\geq 0$,
\begin{equation}
(z_i^2 \partial_{z_i}-r_i\cdot z_i)^k\cdot f(z_1,\ldots,z_d)=(-1)^k\cdot  \mathcal{R}_{z_i,r_i}^+\cdot \partial_{z_i}^k\cdot \mathcal{R}_{z_i,r_i}^+\ f(z_1,\ldots,z_d), \label{remark2.1:eq1}
\end{equation}
where the operator $\mathcal{R}_{z_i,r_i}^+$ is defined in \eqref{notation:eq1}. 
The above identity will be used in Section \ref{section4.1}.
\end{remark}

\begin{remark}
Let $\mathbf{A}$ be a matrix in $\mathbb{R}^{n\times d}$. Then Proposition \ref{prop21} {\rm (iii)} shows that  the expected polynomial has only nonnegative real roots when $k\leq\text{\rm rank}(\mathbf{A})$, and its largest root is non-increasing as $k$ increases. Moreover, Proposition \ref{prop21} {\rm (i)} shows that the expected polynomial $(x^2\cdot\partial_x-d\cdot x)^k \ \det[x\cdot \mathbf{I}_d-\mathbf{A}^{\rm T}\mathbf{A}]$ can be written as $(-1)^k\cdot \mathcal{R}_{x,d}^+\cdot \partial_x^k\cdot \mathcal{R}_{x,d}^+\  \det[x\cdot\mathbf{I}_d-\mathbf{A}^{\rm T}\mathbf{A}]$.  
It is worth noting that Marcus-Spielman-Srivastava and Ravichandran previously used the method of interlacing polynomials to study the restricted invertibility problem \cite{ravi1,inter3}. In the setting of selecting $d-k$ columns from a given matrix $\mathbf{A}\in\mathbb{R}^{n\times d}$, the relevant expected polynomial is $\partial_x^k\  \det[x\cdot\mathbf{I}_d-\mathbf{A}^{\rm T}\mathbf{A}]$. 
Therefore, our expected polynomial can be regarded as an extension of their expected polynomial by incorporating two extra variable substitution operators $\mathcal{R}_{x,d}^+$.

\end{remark}

\subsection{Alternative expressions for the polynomial $p_S(x)$}\label{section4.1}

Let $\mathbf{Z}:=\text{\rm diag}(z_1,\ldots,z_d)$, where $z_1,\ldots,z_d$ are variables. For any matrix $\mathbf{B}\in\mathbb{R}^{d\times d}$, we have (see  \cite[Corollary 3]{ravi2})
\begin{equation}\label{section4.1:eq1}
\det[\mathbf{Z}-\mathbf{B}]=\sum\limits_{R\subset[d]}(-1)^{d-\vert R\vert}\det[\mathbf{B}(R^C)]\cdot z^R.	
\end{equation}
Here  $z^R:=\prod_{j\in R}z_j$ for each subset $R\subset[d]$. For any subset $S\subset[d]$, differentiating both sides of \eqref{section4.1:eq1} with respect to $z_j$ for $j\in S$, we obtain the following identity (See also \cite[Lemma 1]{ravi2}):
\begin{equation}\label{section4.1:eq2}
\det[(\mathbf{Z}-\mathbf{B})(S^C)]=\bigg(\prod\limits_{j\in S}\partial_{z_j}\bigg)\ \det[\mathbf{Z}-\mathbf{B}].
\end{equation}
This shows that the characteristic polynomial of the principal submatrix $\mathbf{B}(S^C)$ can be expressed in terms of the multi-affine polynomial $\det[\mathbf{Z}-\mathbf{B}]$:
\begin{equation*}
\det[x\cdot \mathbf{I}_{d-\vert S\vert}-\mathbf{B}(S^C)]=\bigg(\prod\limits_{j\in S}\partial_{z_j}\bigg)\ \det[\mathbf{Z}-\mathbf{B}]\ \bigg|_{z_j=x,\ \forall j\in[d]}.
\end{equation*}
In this subsection, we provide an analog result for the polynomial $p_S(x)=\det[x\cdot \mathbf{I}_d-\mathbf{A}^{\rm T}\mathbf{Q}_S\mathbf{A}]$ defined in \eqref{section1.3:eq1}, where $\mathbf{Q}_S=\mathbf{I}_n-\mathbf{A}_S\mathbf{A}_S^{\dagger}$.

\begin{prop}\label{prop4.1}

Let $\mathbf{A}=[\mathbf{a}_1,\ldots,\mathbf{a}_d]$ be a matrix in $\mathbb{R}^{n\times d}$. Let $\mathbf{Z}=\text{\rm diag}(z_1,\ldots,z_d)$, where $z_1,\ldots,z_d$ are variables.

\begin{enumerate} 
	\item[(i)] For each subset $S\subset[d]$, we have
\begin{equation}\label{prop4.1:eq1}
\det[\mathbf{A}_S^{\rm T}\mathbf{A}_S]\cdot \det[\mathbf{Z}-\mathbf{A}^{\rm T} \mathbf{Q}_S\mathbf{A}]=\bigg(\prod\limits_{j\in S}^{}(z_j^2\partial_{z_j}-z_j)\bigg)  \det[\mathbf{Z}-\mathbf{A}^{\rm T}\mathbf{A}]
\end{equation}
and
\begin{equation}\label{prop4.1:eq11}
\det[\mathbf{A}_S^{\rm T}\mathbf{A}_S]\cdot p_S(x)=\bigg(\prod\limits_{j\in S}^{}(z_j^2\partial_{z_j}-z_j)\bigg)  \det[\mathbf{Z}-\mathbf{A}^{\rm T}\mathbf{A}]\ \bigg|_{z_j=x,\forall j\in[d]}.
\end{equation}

\item[(ii)] Assume that $\mathbf{A}$ has rank $d$. Let $\widehat{\mathbf{A}}:=(\mathbf{A}^{\rm T}\mathbf{A})^{-1}\in\mathbb{R}^{d\times d}$. Then, for each subset $S\subset[d]$, we have
\begin{equation}\label{prop4.1:eq3}
p_S(x)=(-1)^{\vert S\vert+d}\cdot \frac{\det[\mathbf{A}^{\rm T}\mathbf{A}]}{\det[\mathbf{A}_S^{\rm T}\mathbf{A}_S]}\cdot \mathcal{R}_{x,d}^+\ \det[x\cdot \mathbf{I}_{d-|S|}-\widehat{\mathbf{A}}(S^C)].
\end{equation}

\end{enumerate}

\end{prop}

\begin{proof}

(i) We first prove equation \eqref{prop4.1:eq1}. For any multi-affine polynomial $f(z_1,\ldots,z_d)$ and any $j\in[d]$, we have
\begin{equation*}
(z_j^2\partial_{z_j}-z_j)\ f(z_1,\ldots,z_d)= f(z_1,\ldots,z_d)\ \bigg|_{z_j=0}\times (-z_j).
\end{equation*}
Also, by \eqref{section4.1:eq1}, we can expand out $\det[\mathbf{Z}-\mathbf{A}^{\rm T}\mathbf{A}]$ as
\begin{equation*}
\det[\mathbf{Z}-\mathbf{A}^{\rm T}\mathbf{A}]=\sum\limits_{W\subset[d]}(-1)^{\vert W\vert} \cdot\det[\mathbf{A}_W^{\rm T}\mathbf{A}_W]\cdot z^{W^C}.
\end{equation*}	
Then we can rewrite the right-hand side of \eqref{prop4.1:eq1} as
\begin{equation}\label{prop4.1:eq4}
\begin{aligned}
\bigg(\prod\limits_{j\in S}^{}(z_j^2\partial_{z_j}-z_j)\bigg)  \det[\mathbf{Z}-\mathbf{A}^{\rm T}\mathbf{A}]&=(-1)^{\vert S\vert}\cdot  \det[\mathbf{Z}-\mathbf{A}^{\rm T}\mathbf{A}]\bigg|_{z_j=0,\forall j\in S}\cdot z^{S}\\
&=(-1)^{\vert S\vert} \cdot z^{S} \cdot \sum\limits_{W:S\subset W \subset[d]}(-1)^{\vert W\vert} \cdot\det[\mathbf{A}_W^{\rm T}\mathbf{A}_W]\cdot z^{W^C}\\
&=\sum\limits_{R\subset[d]\backslash S}(-1)^{\vert R\vert} \cdot\det[\mathbf{A}_{R\cup S}^{\rm T}\mathbf{A}_{R\cup S}]\cdot z^{R^C}.
\end{aligned}
\end{equation}
We now derive \eqref{prop4.1:eq1}. By \eqref{section4.1:eq1}, we can rewrite the left-hand side of \eqref{prop4.1:eq1} as
\begin{equation}\label{prop4.1:eq5}
\det[\mathbf{A}_S^{\rm T}\mathbf{A}_S]\cdot \det[\mathbf{Z}-\mathbf{A}^{\rm T} \mathbf{Q}_S\mathbf{A}]=\det[\mathbf{A}_S^{\rm T}\mathbf{A}_S]\cdot \sum\limits_{R\subset[d]}(-1)^{\vert R\vert} \cdot\det[\mathbf{A}_R^{\rm T}\mathbf{Q}_S\mathbf{A}_R]\cdot z^{R^C}.
\end{equation}
The columns of $\mathbf{Q}_S\mathbf{A}_R$ indexed by $R\cap S$ are all equal to $\mathbf{0}_{n\times 1}$, so we have $\det[\mathbf{A}_R^{\rm T}\mathbf{Q}_S\mathbf{A}_R]=0$ if $R\cap S\neq \emptyset$. This means that we can rewrite equation \eqref{prop4.1:eq5} as
\begin{equation*}
\begin{aligned}
\det[\mathbf{A}_S^{\rm T}\mathbf{A}_S]\cdot \det[\mathbf{Z}-\mathbf{A}^{\rm T} \mathbf{Q}_S\mathbf{A}]&=\det[\mathbf{A}_S^{\rm T}\mathbf{A}_S]\cdot \sum\limits_{R\subset[d]\backslash S}(-1)^{\vert R\vert} \cdot\det[\mathbf{A}_R^{\rm T}\mathbf{Q}_S\mathbf{A}_R]\cdot z^{R^C}\\
&\overset{(a)}{=}\sum\limits_{R\subset[d]\backslash S}(-1)^{\vert R\vert} \cdot\det[\mathbf{A}_{R\cup S}^{\rm T}\mathbf{A}_{R\cup S}]\cdot z^{R^C}\\
&=\bigg(\prod\limits_{j\in S}^{}(z_j^2\partial_{z_j}-z_j)\bigg)  \det[\mathbf{Z}-\mathbf{A}^{\rm T}\mathbf{A}].
\end{aligned}
\end{equation*}
Here, the last equality follows from \eqref{prop4.1:eq4}, and the equality $(a)$ follows from
\begin{equation*}
\det[\mathbf{A}_{R\cup S}^{\rm T}\mathbf{A}_{R\cup S}]=	\det[\mathbf{A}_S^{\rm T}\mathbf{A}_S]\cdot \det[\mathbf{A}_R^{\rm T}\mathbf{Q}_S\mathbf{A}_R], \quad \forall\ R\subset[d]\backslash S,
\end{equation*}
which is obtained by taking $\mathbf{B}=\mathbf{A}_S$ and $\mathbf{C}=\mathbf{A}_R$ in Lemma \ref{lemma2.2}.
Hence, we arrive at \eqref{prop4.1:eq1}. Setting $z_j=x$ for each $j\in[d]$ in equation \eqref{prop4.1:eq1}, we immediately obtain equation \eqref{prop4.1:eq11}.

(ii) A simple observation is that, for any multi-affine polynomial $f(z_1,\ldots,z_d)$ and any $j\in[d]$, we have $\mathcal{R}_{z_j,1}^+\cdot \mathcal{R}_{z_j,1}^+\ f= f$. Combining this with \eqref{remark2.1:eq1} we have
\begin{equation}\label{prop4.1:eq7}
\begin{aligned}
\bigg(\prod\limits_{j\in S}^{}(z_j^2\partial_{z_j}-z_j)\bigg)  \det[\mathbf{Z}-\mathbf{A}^{\rm T}\mathbf{A}]&=\bigg(\prod\limits_{j\notin S}^{}  \mathcal{R}_{z_j,1}^+\cdot \mathcal{R}_{z_j,1}^+\bigg)\cdot  \bigg(\prod\limits_{j\in  S}^{} -\mathcal{R}_{z_j,1}^+\cdot \partial_{z_j}\cdot \mathcal{R}_{z_j,1}^+\bigg)\cdot  \det[\mathbf{Z}-\mathbf{A}^{\rm T}\mathbf{A}] \\
&=(-1)^{\vert S\vert}\cdot \mathcal{R}\cdot \bigg( \prod\limits_{j\in S}^{} \partial_{z_j}\bigg)\cdot \mathcal{R} \cdot  \det[\mathbf{Z}-\mathbf{A}^{\rm T}\mathbf{A}],
\end{aligned}
\end{equation}
where $\mathcal{R}:=\prod_{j=1}^{d} \mathcal{R}_{z_j,1}^+$, and the last equality follows from the fact that the operator $\mathcal{R}_{z_j,1}^+$ can be interchanged with the operators $\mathcal{R}_{z_i,1}^+$ and $\partial_{z_i}$ for any $i\neq j$.	
Note that
\begin{equation}\label{prop4.1:eq8}
\mathcal{R}\cdot  \det[\mathbf{Z}-\mathbf{A}^{\rm T}\mathbf{A}]=\det[\mathbf{Z}]\cdot \det[\text{\rm diag}({z_1}^{-1},\ldots,{z_d}^{-1})-\mathbf{A}^{\rm T}\mathbf{A}]=(-1)^d\cdot \det[\mathbf{A}^{\rm T}\mathbf{A}]\cdot \det[\mathbf{Z}-\widehat{\mathbf{A}}].
\end{equation}
Then combining \eqref{prop4.1:eq11}, \eqref{prop4.1:eq7} and \eqref{prop4.1:eq8}, we obtain
\begin{equation*}
\begin{aligned}
\det[\mathbf{A}_S^{\rm T}\mathbf{A}_S]\cdot p_S(x)&=\bigg(\prod\limits_{j\in S}^{}(z_j^2\partial_{z_j}-z_j)\bigg)  \det[\mathbf{Z}-\mathbf{A}^{\rm T}\mathbf{A}]\ \bigg|_{z_j=x,\forall j}\\
&=(-1)^{\vert S\vert+d}\cdot \det[\mathbf{A}^{\rm T}\mathbf{A}]\cdot   \mathcal{R}  \cdot \bigg( \prod\limits_{j\in S}^{} \partial_{z_j}\bigg) \det[\mathbf{Z}-\widehat{\mathbf{A}}]\ \bigg|_{z_j=x,\forall j}\\
&\overset{(a)}=(-1)^{\vert S\vert+d}\cdot \det[\mathbf{A}^{\rm T}\mathbf{A}]\cdot   \mathcal{R} \cdot  \det[(\mathbf{Z}-\widehat{\mathbf{A}})(S^C)]\ \bigg|_{z_j=x,\forall j}\\
&=(-1)^{\vert S\vert+d}\cdot \det[\mathbf{A}^{\rm T}\mathbf{A}]\cdot \mathcal{R}_{x,d}^+\ \det[x\cdot \mathbf{I}_{d-|S|}-\widehat{\mathbf{A}}(S^C)],	
\end{aligned}
\end{equation*}
where (a) follows from \eqref{section4.1:eq2}. In viewing of $\det[\mathbf{A}_S^{\rm T}\mathbf{A}_S]>0$, we arrive at the conclusion.
	
\end{proof}

\begin{remark}\label{remark4.1}
Consider the case where $\mathbf{A}\in\mathbb{R}^{n\times d}$ is a rank-$d$ matrix. By utilizing \eqref{prop4.1:eq3}, we can demonstrate the connection between Problem \ref{problem1.1} and Problem \ref{problem1.3}.
Indeed, we define $\mathbf{B}=(\mathbf{A}^{\dag})^{\rm T}$. Then, we have $\widehat{\mathbf{A}}=(\mathbf{A}^{\rm T}\mathbf{A})^{-1}=\mathbf{B}^{\rm T}\mathbf{B}$.
Since the polynomial $p_S(x)$ has exactly $d-|S|$ positive roots for each  $S\subset [d]$, equation \eqref{prop4.1:eq3} implies that
\begin{equation}\label{remark4.1:eq1}
\text{\rm maxroot}\ p_S(x)=	\frac{1}{\lambda_{\rm min}(\widehat{\mathbf{A}}(S^C))}=\frac{1}{\sigma_{\rm min}^2(\mathbf{B}_{S^C})}.
\end{equation}
 By combining it with \eqref{section3.1:eq2}, we observe that the task of minimizing the residual $\Vert \mathbf{A}-\mathbf{A}_S\mathbf{A}_S^{\dagger}\mathbf{A} \Vert_2^2$ is tantamount to maximizing $\sigma_{\rm min}(\mathbf{B}{(S^C)})$, where $S$ ranges over all $k$-subsets of $[d]$. Therefore, when $\text{\rm rank}(\mathbf{A})=d$, Problem \ref{problem1.1} can be transformed into Problem \ref{problem1.3} by choosing $\mathbf{B}=(\mathbf{A}^{\dag})^{\rm T}$ and $l=d-k$.
 It is worth mentioning that in the general case, when matrix $\mathbf{A}$ does not have full column rank, it may not be feasible to find a matrix $\widehat{\mathbf{A}}$ that satisfies equation  \eqref{prop4.1:eq3} for every subset $S$. As a result, 
 when $\mathrm{rank}(\mathbf{A})<d$, we cannot rely on the previous findings on restricted invertibility to obtain the bound \eqref{thm1.1:eq1}  stated in Theorem \ref{thm1.1}.

\end{remark}

\subsection{Real-rootedness, interlacing, and real stability}\label{section24}

Here we introduce the definition of common interlacing.

\begin{definition}\label{def2.1}
Let $f(x)=a_0\cdot \prod_{i=1}^{d-1}(x-a_i)$ and $p(x)=b_0\cdot \prod_{i=1}^{d}(x-b_i)$ be two real-rooted polynomials. We say that $f(x)$ interlaces $p(x)$ if
\begin{equation*}
b_1\leq a_1\leq b_2\leq a_2\leq \cdots \leq a_{d-1}\leq b_{d}.
\end{equation*}
We say that real-rooted polynomials $p_1(x),\ldots,p_m(x)$ have a common interlacing if there exists a polynomial $f(x)$ such that $f(x)$ interlaces $p_i(x)$ for each $i\in[m]$.
\end{definition}

By Definition \ref{def2.1}, to prove that the polynomials $p_1(x),\ldots,p_m(x)$ have a common interlacing, it is enough to show that any two of them have a common interlacing. By a result of Fell \cite[Theorem 2]{Fell}, this is equivalent to proving that the polynomial $\mu\cdot p_i(x)+(1-\mu)\cdot p_j(x)$ is real-rooted for any $i,j\in[m]$ and for any $\mu\in[0,1]$.

The following lemma provides a sufficient condition under which a collection of polynomials have a common interlacing.

\begin{lemma}{\rm \cite[Claim 2.9]{inter3}}\label{lemma2.5}
Assume that $\mathbf{A}\in\mathbb{R}^{d\times d}$ is a symmetric matrix and assume $\mathbf{b}_1,\ldots,\mathbf{b}_m$ are vectors in $\mathbb{R}^d$. Then the polynomials
\begin{equation*}
p_i(x)=\det[x\cdot\mathbf{I}_d-\mathbf{A}+\mathbf{b}_i\mathbf{b}_i^{\rm T}], \quad i=1,\ldots,m
\end{equation*}
have a common interlacing.

\end{lemma}

Marcus-Spielman-Srivastava proved that when a set of polynomials have a common interlacing, one of these polynomials has its largest root upper bounded by the largest root of the sum of all the polynomials in the set:

\begin{lemma}\label{lemma2.6}{\rm\cite[Lemma 4.2]{inter1}}
Assume that $p_1(x),\ldots,p_m(x)$ are real-rooted polynomials with the same degree and positive leading coefficients. If $p_1(x),\ldots,p_m(x)$ have a common interlacing, then there exists an integer $i\in[m]$ such that
\begin{equation*}
\text{\rm maxroot}\ p_i(x)\leq \text{\rm maxroot}\ \sum\limits_{i=1}^{m}p_i(x).
\end{equation*}
\end{lemma}

\section{Proof of Theorem \ref{thm1.1} and A Deterministic Polynomial-Time Algorithm}\label{section3}

\subsection{The method of interlacing polynomials: proof of Theorem \ref{thm1.3}}

Let $\mathbf{A}=[\mathbf{a}_1,\ldots,\mathbf{a}_d]$ be a matrix in $\mathbb{R}^{n\times d}$. For each subset $S\subset[d]$, we denote $\mathbf{Q}_S$ by the projection matrix of the orthogonal complement of $\mathrm{Ran}(\mathbf{A}_S)\equiv \text{\rm span}\{\mathbf{a}_i,i\in S\}$, i.e., $\mathbf{Q}_S:=\mathbf{I}_n-\mathbf{A}_S\mathbf{A}_S^{\dagger}$. Since $\mathbf{Q}_S$ is a projection matrix, we have $\mathbf{Q}_S=\mathbf{Q}_S^{\rm T}=\mathbf{Q}_S^2$.
For each subset $S\subset[d]$, we define the degree $d$ polynomial
\begin{equation}\label{section3.1:eq1}
p_S(x):=\det[x\cdot\mathbf{I}_d-(\mathbf{A}-\mathbf{A}_{S}\mathbf{A}_{S}^{\dagger}\mathbf{A})^{\rm T}(\mathbf{A}-\mathbf{A}_{S}\mathbf{A}_{S}^{\dagger}\mathbf{A})]=\det[x\cdot\mathbf{I}_d-\mathbf{A}^{\rm T}\mathbf{Q}_{S}\mathbf{A}].
\end{equation}
Since $\mathbf{A}^{\rm T}\mathbf{Q}_{S}\mathbf{A}$ is a symmetric positive semidefinite matrix, the polynomial $p_S(x)$ is in $\mathbb{P}^+(d)$ for each $S\subset[d]$. A simple observation is that
\begin{equation*}
\text{\rm maxroot}\ p_S(x)=\Vert \mathbf{A}-\mathbf{A}_{S}\mathbf{A}_{S}^{\dagger}\mathbf{A}\Vert_2^2.
\end{equation*}
Noting that $\mathbf{Q}_S\mathbf{a}_j=\mathbf{0}_{n}$ for each $\mathbf{a}_j\in \text{\rm span}\{\mathbf{a}_i : i\in S\}$,  we have $\text{\rm rank}(\mathbf{Q}_S\mathbf{A})=\text{\rm rank}(\mathbf{A})-\text{\rm rank}(\mathbf{A}_S)$, which implies that the polynomial $p_S(x)$ has exactly $\text{\rm rank}(\mathbf{A})-\text{\rm rank}(\mathbf{A}_S)$ positive roots. In particular, if $S=\emptyset$, then $p_S(x)=\det[x\cdot\mathbf{I}_d-\mathbf{A}^{\rm T}\mathbf{A}]$. If  $\text{\rm rank}(\mathbf{A}_S)=\text{\rm rank}(\mathbf{A})$, where $S$ is a subset of $[d]$, then $p_S(x)=x^d$.

To prove Theorem \ref{thm1.3}, we introduce some lemmas. The first lemma shows that the polynomial $(x^2\cdot\partial_x-d\cdot x)\ p_{S}(x)$ is a weighted sum of the polynomials $p_{S\cup\{i\}}(x)$ for all $i\notin S$ with nonnegative coefficients.

\begin{lemma}\label{lemma3.1}
Let $\mathbf{A}=[\mathbf{a}_1,\ldots,\mathbf{a}_d]$ be a matrix in $\mathbb{R}^{n\times d}$. For every subset $S\subset[d]$, we have
\begin{equation*}
\sum\limits_{i\notin S}\Vert \mathbf{Q}_S\mathbf{a}_i\Vert^2\cdot p_{S\cup\{i\}}(x)= (x^2\cdot\partial_x-d\cdot x)\ p_S(x),
\end{equation*}
where $\mathbf{Q}_S=\mathbf{I}_n-\mathbf{A}_S\mathbf{A}_S^\dagger$ and $p_S(x)=\det [x\cdot \mathbf{I}_d-\mathbf{A}^{\rm T}\mathbf{Q}_S\mathbf{A}]$.
\end{lemma}

\begin{proof}

Since $\mathbf{Q}_S$ is the projection matrix of  $\text{\rm span}\{\mathbf{a}_j,j\in S\}^\bot$,
it is enough to prove that
\begin{equation}\label{lemma3.1:eq1}
\sum\limits_{i:\Vert \mathbf{Q}_S\mathbf{a}_i\Vert> 0}\Vert \mathbf{Q}_S\mathbf{a}_i\Vert^2\cdot p_{S\cup\{i\}}(x)= (x^2\cdot\partial_x-d\cdot x)\ p_S(x).
\end{equation}
When $\Vert \mathbf{Q}_S\mathbf{a}_i\Vert>0$, i,e., $\mathbf{a}_i\notin \text{\rm span}\{\mathbf{a}_j,j\in S\}$, Lemma \ref{lemma2.1} implies
\begin{equation*}
\mathbf{Q}_{S\cup\{i\}}=\mathbf{I}_n-\mathbf{A}_{S\cup\{i\}}\mathbf{A}_{S\cup\{i\}}^{\dagger}=\mathbf{I}_n-\left(\mathbf{A}_{S}\mathbf{A}_{S}^{\dagger}+\frac{\mathbf{Q}_S\mathbf{a}_i \mathbf{a}_i^{\rm T} \mathbf{Q}_S }{ \Vert \mathbf{Q}_S\mathbf{a}_i\Vert^2}\right)=\mathbf{Q}_{S}-\frac{\mathbf{Q}_S\mathbf{a}_i \mathbf{a}_i^{\rm T} \mathbf{Q}_S }{ \Vert \mathbf{Q}_S\mathbf{a}_i\Vert^2}.
\end{equation*}
Therefore, we can write $p_{S\cup\{i\}}(x)$ as
\begin{equation}\label{lemma3.1:eq2}
p_{S\cup\{i\}}(x)=\det[x\cdot\mathbf{I}_d-\mathbf{A}^{\rm T}\mathbf{Q}_{S\cup\{i\}}\mathbf{A}]=\det[x\cdot\mathbf{I}_d-\mathbf{A}^{\rm T}\mathbf{Q}_{S}\mathbf{A}+\mathbf{A}^{\rm T}\frac{\mathbf{Q}_S\mathbf{a}_i \mathbf{a}_i^{\rm T} \mathbf{Q}_S }{ \Vert \mathbf{Q}_S\mathbf{a}_i\Vert^2}\mathbf{A}].
\end{equation}
We claim that for any matrix $\mathbf{B}=[\mathbf{b}_1,\ldots,\mathbf{b}_d]\in\mathbb{R}^{n\times d}$,
\begin{equation}\label{lemma2.4:eq1}
(x^2\cdot\partial_x-d\cdot x)\cdot\det[x\cdot \mathbf{I}_d-\mathbf{B}^{\rm T}\mathbf{B}]=\sum\limits_{i:\Vert \mathbf{b}_i\Vert>0}\Vert \mathbf{b}_i\Vert^2 \cdot\det[x\cdot \mathbf{I}_d-\mathbf{B}^{\rm T}(\mathbf{I}_n-\frac{\mathbf{b}_i\mathbf{b}_i^{\rm T}}{\Vert \mathbf{b}_i\Vert^2 })\mathbf{B}].
\end{equation}
Since $\mathbf{Q}_S=\mathbf{Q}_S^2=\mathbf{Q}_S^{\rm T}$, we have $(\mathbf{Q}_S\mathbf{A})^{\rm T}\mathbf{Q}_S\mathbf{A}=\mathbf{A}^{\rm T}\mathbf{Q}_S\mathbf{A}$.
Taking $\mathbf{B}=\mathbf{Q}_S\mathbf{A}=[\mathbf{Q}_S\mathbf{a}_1,\ldots,\mathbf{Q}_S\mathbf{a}_d]$ in \eqref{lemma2.4:eq1}, we obtain that
\begin{equation*}
\begin{aligned}
(x^2\cdot\partial_x-d\cdot x)\cdot\det[x\cdot \mathbf{I}_d-\mathbf{A}^{\rm T}\mathbf{Q}_S\mathbf{A} ]
&=\sum\limits_{i:\Vert  \mathbf{Q}_S\mathbf{a}_i\Vert>0}\Vert \mathbf{Q}_S\mathbf{a}_i\Vert^2 \cdot\det[x\cdot \mathbf{I}_d-\mathbf{A}^{\rm T}\mathbf{Q}_S\mathbf{A}+\mathbf{A}^{\rm T}\frac{\mathbf{Q}_S\mathbf{a}_i \mathbf{a}_i^{\rm T} \mathbf{Q}_S }{ \Vert \mathbf{Q}_S\mathbf{a}_i\Vert^2}\mathbf{A}]\\
&=\sum\limits_{i:\Vert  \mathbf{Q}_S\mathbf{a}_i\Vert>0}\Vert \mathbf{Q}_S\mathbf{a}_i\Vert^2 \cdot p_{S\cup\{i\}}(x).
\end{aligned}
\end{equation*}
Since $p_S(x)=\det[x\cdot \mathbf{I}_d-\mathbf{A}^{\rm T}\mathbf{Q}_S\mathbf{A}]$, we arrive at  \eqref{lemma3.1:eq1}.

It remains to prove \eqref{lemma2.4:eq1}. By the Leibniz formula, for each multi-affine polynomial $f(z_{1},\ldots,z_{d})$, we have the following algebraic identity:
\begin{equation*}
(x^2\cdot\partial_{x}-d\cdot x)\cdot f(x,\ldots,x)=\sum\limits_{i=1}^{d}(z_{i}^2\cdot\partial_{z_{i}}-z_{i})\cdot  f(z_{1},\ldots,z_{d})\ \bigg|_{z_{i}=x,\forall i}.
\end{equation*}
Substituting $f(z_1,\ldots,z_d)=\det[\mathbf{Z}-\mathbf{B}^{\rm T}\mathbf{B}]$ with $\mathbf{Z}=\text{\rm diag}(z_1,\ldots,z_d)$ into the above identity, we obtain that
\begin{equation*}
\begin{aligned}
(x^2\cdot\partial_x-d\cdot x)\cdot\det[x\cdot \mathbf{I}_d-\mathbf{B}^{\rm T}\mathbf{B}]
&=\sum\limits_{i=1}^{d}(z_{i}^2\cdot\partial_{z_{i}}-z_{i})\cdot  \det[\mathbf{Z}-\mathbf{B}^{\rm T}\mathbf{B}]\ \bigg|_{z_{i}=x,\forall i}\\
&\overset{(a)}=\sum\limits_{i=1}^{d} \Vert \mathbf{b}_i\Vert^2 \cdot\det[\mathbf{Z}-\mathbf{B}^{\rm T}(\mathbf{I}_n-\mathbf{b}_i\mathbf{b}_i^{\dagger})\mathbf{B}] \bigg|_{z_{i}=x,\forall i}\\
&\overset{(b)}=\sum\limits_{i:\Vert \mathbf{b}_i\Vert>0}\Vert \mathbf{b}_i\Vert^2 \cdot\det[x\cdot \mathbf{I}_d-\mathbf{B}^{\rm T}(\mathbf{I}_n-\frac{\mathbf{b}_i\mathbf{b}_i^{\rm T}}{\Vert \mathbf{b}_i\Vert^2 })\mathbf{B}],
\end{aligned}
\end{equation*}
where (a) follows from Proposition \ref{prop4.1} (i) and (b) follows from the fact that $\mathbf{b}^{\dagger}=\frac{\mathbf{b}^{\rm T}}{\Vert \mathbf{b}\Vert^2}$ for any nonzero vector $\mathbf{b}$. Hence we arrive at \eqref{lemma2.4:eq1}.
\end{proof}

Next, we will employ the method of interlacing polynomials to prove the following lemma.
\begin{lemma}\label{lemma3.2}

Let $\mathbf{A}=[\mathbf{a}_1,\ldots,\mathbf{a}_d]$ be a matrix in $\mathbb{R}^{n\times d}$, and let $S\subset [d]$ be such that the set $R:=\{i: \Vert \mathbf{Q}_S\mathbf{a}_i \Vert>0\}$ is nonempty. For any nonnegative integer $l\leq\text{\rm rank}(\mathbf{A})-\text{\rm rank}(\mathbf{A}_S)-1$, there exists an integer $j\in R$ such that
\begin{equation*}
\text{\rm maxroot}\ (x^2\cdot\partial_x-d\cdot x)^{l}\cdot  p_{S\cup \{j\}}(x) \leq \text{\rm maxroot}\ (x^2\cdot\partial_x-d\cdot x)^{l+1}\cdot p_S(x).
\end{equation*}

\end{lemma}

\begin{proof}
 We set
\begin{equation}\label{eq:gS}
g_{S\cup\{i\}}^{(l)}(x):=(x^2\cdot\partial_x-d\cdot x)^{l}\cdot \bigg( \Vert \mathbf{Q}_S\mathbf{a}_i\Vert^2\cdot p_{S\cup\{i\}}(x)\bigg)=(-1)^l\cdot \mathcal{R}_{x,d}^+\cdot\partial_x^l\cdot \mathcal{R}_{x,d}^+\cdot  \bigg( \Vert \mathbf{Q}_S\mathbf{a}_i\Vert^2\cdot p_{S\cup\{i\}}(x)\bigg),
\end{equation}
where $i\in R$ and  $l\in [0,\text{\rm rank}(\mathbf{A})-\text{\rm rank}(\mathbf{A}_S)-1]\cap {\mathbb Z}$.
A simple observation is that $\mathrm{rank}(\mathbf{Q}_{S\cup\{i\}}\mathbf{A})=\text{\rm rank}(\mathbf{A})-\text{\rm rank}(\mathbf{A}_S)-1$ for $i\in R$,  so the polynomial $p_{S\cup\{i\}}(x)$ has exactly $\text{\rm rank}(\mathbf{A})-\text{\rm rank}(\mathbf{A}_S)-1$ nonzero roots. According to (iii) in Proposition \ref{prop21}, the polynomial $g_{S\cup\{i\}}^{(l)}(x)$ is in $\mathbb{P}^+(d)$ for each $l\leq\text{\rm rank}(\mathbf{A})-\text{\rm rank}(\mathbf{A}_S)-1$ and for each $i\in R$.

For any nonnegative integer $l\leq\text{\rm rank}(\mathbf{A})-\text{\rm rank}(\mathbf{A}_S)-1$, we claim that the polynomials $\{g_{S\cup\{i\}}^{(l)}(x) :  i\in R\}$ have a common interlacing. Then, by Lemma \ref{lemma2.6}, there exists an integer $j\in R$ such that
\begin{equation}\label{lemma3.2:eq1}
\text{\rm maxroot}\ g^{(l)}_{S\cup \{j\}}(x)\leq \text{\rm maxroot}\ \sum\limits_{i\in R}^{}g^{(l)}_{S\cup \{i\}}(x).
\end{equation}
According to \eqref{eq:gS}, we have
\begin{equation}\label{lemma3.2:eq2}
\text{\rm maxroot}\ g^{(l)}_{S\cup \{j\}}(x)= \text{\rm maxroot}\ (x^2\cdot\partial_x-d\cdot x)^{l}\cdot p_{S\cup \{j\}}(x)
\end{equation}
and
\begin{equation}\label{lemma3.2:eq3}
\sum\limits_{i\in R}^{}g^{(l)}_{S\cup \{i\}}(x)=(x^2\cdot\partial_x-d\cdot x)^{l} \sum\limits_{i\in R}^{}\Vert \mathbf{Q}_S\mathbf{a}_i\Vert^2\cdot p_{S\cup\{i\}}(x)=(x^2\cdot\partial_x-d\cdot x)^{l+1}\ p_S(x),
\end{equation}
where the last equation follows from Lemma \ref{lemma3.1}. Hence, combining \eqref{lemma3.2:eq1}, \eqref{lemma3.2:eq2}, and \eqref{lemma3.2:eq3},  we obtain that
\begin{equation*}
\text{\rm maxroot}\ (x^2\cdot\partial_x-d\cdot x)^{l}\cdot  p_{S\cup \{j\}}(x)\leq \text{\rm maxroot}\ (x^2\cdot\partial_x-d\cdot x)^{l+1}\cdot  p_S(x).
\end{equation*}

It remains to show that, for any nonnegative integer $l\leq\text{\rm rank}(\mathbf{A})-\text{\rm rank}(\mathbf{A}_S)-1$, the polynomials $\{g_{S\cup\{i\}}^{(l)}(x) : i\in R\}$ have a common interlacing. According to \eqref{lemma3.1:eq2}, we obtain that
\begin{equation*}
p_{S\cup\{i\}}(x)=\det[x\cdot\mathbf{I}_d-\mathbf{A}^{\rm T}\mathbf{Q}_{S}\mathbf{A}+\frac{1}{\Vert \mathbf{Q}_S\mathbf{a}_i \Vert^2}\cdot(\mathbf{A}^{\rm T}\mathbf{Q}_S\mathbf{a}_i)(\mathbf{A}^{\rm T}\mathbf{Q}_S\mathbf{a}_i)^{\rm T}],\quad\forall\ i\in R.
\end{equation*}
Then Lemma \ref{lemma2.5} shows that the polynomials $p_{S\cup\{i\}}(x)$ for $i\in R$ have a common interlacing. Note that the polynomials $g^{(0)}_{S\cup\{i\}}(x)$ and $p_{S\cup\{i\}}(x)$ have the same roots for each $i\in R$, so the polynomials $g^{(0)}_{S\cup\{i\}}(x)$ for $i\in R$  have a common interlacing.
Note that the operators $\partial_x$ and $\mathcal{R}_{x,d}^+$ both preserve interlacing for degree $d$ real-rooted polynomials with the same number of nonzero roots. Then, for each positive integer $l\leq\text{\rm rank}(\mathbf{A})-\text{\rm rank}(\mathbf{A}_S)-1$, the polynomials $g^{(l)}_{S\cup\{i\}}(x)$ for $i\in R$ have a common interlacing.

\end{proof}

We have all the materials to prove Theorem \ref{thm1.3}.

\begin{proof}[{{Proof of Theorem \ref{thm1.3}}}]\label{proofthm1.3}

We start with $S=\emptyset$. Note that $p_{\emptyset}(x)=\det[x\cdot \mathbf{I}_d-\mathbf{A}^{\rm T}\mathbf{A}]$. By Lemma \ref{lemma3.2}, there exists an $i_1\in[d]$ such that $\Vert\mathbf{a}_{i_1}\Vert>0$ and
\begin{equation}\label{proof1.3:eq1}
\text{\rm maxroot}\ (x^2\cdot\partial_x-d\cdot x)^{k-1}\cdot  p_{\{i_1\}}(x)\leq \text{\rm maxroot}\ (x^2\cdot\partial_x-d\cdot x)^k\cdot  \det[x\cdot \mathbf{I}_d-\mathbf{A}^{\rm T}\mathbf{A}].
\end{equation}	
Now we set $S=\{i_1\}$. Using Lemma \ref{lemma3.2} again, we obtain that there exists an $i_2\in[d]\backslash\{i_1\}$ such that $\Vert\mathbf{Q}_{\{i_1\}}\mathbf{a}_{i_2}\Vert>0$ and
\begin{equation}\label{proof1.3:eq2}
\text{\rm maxroot}\ (x^2\cdot\partial_x-d\cdot x)^{k-2}\cdot  p_{\{i_1,i_2\}}(x)\leq \text{\rm maxroot}\ (x^2\cdot\partial_x-d\cdot x)^{k-1}\cdot  p_{\{i_1\}}(x) .
\end{equation}	
Combining \eqref{proof1.3:eq1} and \eqref{proof1.3:eq2}, we obtain that
\begin{equation*}
\begin{aligned}
\text{\rm maxroot}\  (x^2\cdot\partial_x-d\cdot x)^{k-2}\cdot p_{\{i_1,i_2\}}(x)&\leq \text{\rm maxroot}\ (x^2\cdot\partial_x-d\cdot x)^{k-1}\ p_{\{i_1\}}(x)\\
&\leq \text{\rm maxroot}\ (x^2\cdot\partial_x-d\cdot x)^k\ \det[x\cdot \mathbf{I}_d-\mathbf{A}^{\rm T}\mathbf{A}].
\end{aligned}
\end{equation*}
Then we set $S=\{i_1,i_2\}$. By repeating the same argument $k-2$ more times, we can obtain a $k$-subset $S=\{i_1,i_2,\ldots,i_k\}\subset[d]$ such that $\text{\rm rank}(\mathbf{A}_S)=k$ and
\begin{equation*}
\text{\rm maxroot}\ p_{S}(x)\leq \text{\rm maxroot}\  (x^2\cdot\partial_x-d\cdot x)\ p_{\{i_1,\ldots,i_{k-1}\}}(x) \leq \cdots \leq \text{\rm maxroot}\ (x^2\cdot\partial_x-d\cdot x)^k \ \det[x\cdot \mathbf{I}_d-\mathbf{A}^{\rm T}\mathbf{A}].
\end{equation*}	

\end{proof}

\subsection{Estimating the largest root of the expected polynomial.}

With the help of Theorem \ref{thm1.3}, to finish the proof of Theorem \ref{thm1.1}, it remains to estimate the largest root of the polynomial $(x^2\cdot\partial_x-d\cdot x)^k\ \det[x\cdot \mathbf{I}_d-\mathbf{A}^{\rm T}\mathbf{A}]$. The main result of this subsection is the following theorem, which provides an upper bound on the largest root of $(x^2\cdot\partial_x-d\cdot x)^k\ p(x)$ for a polynomial $p(x)\in\mathbb{P}^+(d)$. We postpone its proof to the end of this subsection.

\begin{theorem}\label{thm3.1}

Assume that $p(x)=x^{d-t}\prod_{i=1}^{t}(x-\lambda_i)$ is a real-rooted polynomial of degree $d$, where $t\leq d$ is a positive integer and $\lambda_1,\ldots,\lambda_t$ are positive  numbers satisfying $\lambda_1\geq\cdots\geq \lambda_t>0$  and $\lambda_1>\lambda_t$. Set $\alpha:=t/\sum_{i=1}^{t} \lambda_i^{-1}$ and $\beta:=\frac{\lambda_t^{-1}-\alpha^{-1}}{\lambda_t^{-1}-\lambda_1^{-1}}\leq 1$. Then, for each integer $k$ satisfying $\beta\cdot t\leq k< t$, we have
\begin{equation}\label{thm3.1:eq1}
\text{\rm maxroot}\ (x^2\cdot\partial_x-d\cdot x)^k\ p(x)\leq \frac{\lambda_1}{1+(\frac{\lambda_1}{\alpha}-1) \cdot\bigg(\sqrt{\frac{k}{t}}-\sqrt{\frac{\beta}{1-\beta}\cdot (1-\frac{k}{t})} \bigg)^2}.
\end{equation}
\end{theorem}

\begin{remark}
As $k$ increases from $\beta\cdot t$ to $t$, the upper bound in \eqref{thm3.1:eq1} decreases from $\lambda_1$ to $\alpha$.
Theorem \ref{thm3.1} requires  $p(x)$ has at least two distinct nonzero roots, i,e., $\lambda_1>\lambda_t$.
For the case where $\lambda_1=\lambda_t=\lambda>0$, we can write $p(x)$ in the form of  $p(x)=x^{d-t}\cdot (x-\lambda)^t$. Then we use \eqref{prop2.1:eq1} to obtain that
\begin{equation}
(x^2\cdot \partial_x-d\cdot x)^{k}\ p(x)=\frac{ t!}{(t-k)!}\cdot \lambda^k\cdot x^{d+k-t}\cdot (x-\lambda)^{t-k}, \quad \forall\ k\leq t,
\end{equation}
which implies  that $\lambda$ is always the largest root of $(x^2\cdot \partial_x-d\cdot x)^{k}\ p(x)$ for each $k<t$. 	
\end{remark}

To prove Theorem \ref{thm3.1}, we need the following lemma, which shows how the largest root of a univariate polynomial shrink after taking derivatives. Lemma \ref{lemma3.3} was proved using the barrier method introduced by Batson-Spielman-Srivastava in \cite{inter0}. We refer to \cite{inter3,inter5,bcms19,ravi1,Branden2,ravi3,coh2016,XXZ21} for more details about the barrier method and its applications.

\begin{lemma}{\rm \cite[Theorem 4.2]{ravi1}}\label{lemma3.3}
Assume that $p(x)=\prod_{i=1}^{t}(x-\lambda_i)$ is a real-rooted polynomial of degree $t$, where $\lambda_i\in[0,1]$ for each $i\in[t]$. Let $\gamma=\frac{1}{t}\sum_{i=1}^t\lambda_i$. Then, for each $k\geq \gamma\cdot t$,
\begin{equation*}
\text{\rm maxroot}\ \partial_x^k\ p(x)\leq \left(\sqrt{\gamma\cdot\frac{k}{t}}+\sqrt{(1-\gamma)\cdot\left(1-\frac{k}{t}\right)}\right)^2.
\end{equation*}

\end{lemma}

Now we can present the proof of Theorem \ref{thm3.1}.

\begin{proof}[{{Proof of Theorem \ref{thm3.1}}}]\label{proof thm3.1}
We define the polynomials
\begin{equation*}
q(x):=\mathcal{R}_{x,d}^+\ p(x)\quad\text{\rm and}\quad h(x):=q\bigg(\lambda_t^{-1}-(\lambda_t^{-1}-\lambda_1^{-1})\cdot x\bigg)=\lambda_1\cdots\lambda_t\cdot  (\lambda_t^{-1}-\lambda_1^{-1})^t\cdot  \prod\limits_{i=1}^t(x-b_i),
\end{equation*}
where $b_i:=\frac{\lambda_t^{-1}-\lambda_i^{-1}}{\lambda_t^{-1}-\lambda_1^{-1}}\in[0,1]$ for each $i\in[t]$.
By the chain rule of differentiation, we have
\begin{equation}\label{thm3.1:eq2}
\text{\rm minroot}\ \partial_x^k\ q(x)=\lambda_t^{-1}-(\lambda_t^{-1}-\lambda_1^{-1})\cdot \text{\rm maxroot}\ \partial_x^k\ h(x).
\end{equation}
Combining \eqref{thm3.1:eq2} and Proposition \ref{prop21} (iii), for each $k< t$, we have
\begin{equation}\label{thm3.1:eq3}
\text{\rm maxroot}\ (x^2\cdot\partial_x-d\cdot x)^k\ p(x)=\frac{1}{\text{\rm minroot}\ \partial_x^k\ q(x)}=\frac{1}{\lambda_t^{-1}-(\lambda_t^{-1}-\lambda_1^{-1})\cdot \text{\rm maxroot}\ \partial_x^k\ h(x)}.
\end{equation}
Note that all the roots of $h(x)$ are between 0 and 1. Combining  \eqref{thm3.1:eq3} and Lemma \ref{lemma3.3}, for each $k\geq \sum_{i=1}^{t}b_i=\beta\cdot t$ we have
\begin{equation*}
\begin{aligned}
\text{\rm maxroot}\ (x^2\cdot\partial_x-d\cdot x)^k\ p(x)&\leq \frac{1}{\lambda_t^{-1}-(\lambda_t^{-1}-\lambda_1^{-1})\cdot\bigg(\sqrt{\beta \cdot \frac{k}{t}}+\sqrt{(1-\beta)\cdot (1-\frac{k}{t})} \bigg)^2}\\
&=\frac{1}{\lambda_1^{-1}+(\lambda_t^{-1}-\lambda_1^{-1})\cdot\bigg(\sqrt{(1-\beta)\cdot \frac{k}{t}}-\sqrt{\beta\cdot (1-\frac{k}{t})} \bigg)^2}\\
&= \frac{\lambda_1}{1+(\frac{\lambda_1}{\alpha}-1) \cdot\bigg(\sqrt{\frac{k}{t}}-\sqrt{\frac{\beta}{1-\beta}\cdot (1-\frac{k}{t})} \bigg)^2}.
\end{aligned}
\end{equation*}
\end{proof}

\subsection{Proof of Theorem \ref{thm1.1}}\label{proof-Th1}

The proof of Theorem \ref{thm1.1} is presented in this subsection.

\begin{proof}[{{Proof of Theorem \ref{thm1.1}}}]\label{proofthm1.1}

For each subset $S\subset[d]$, we define the degree-$d$ polynomial $p_S(x)$ as in \eqref{section3.1:eq1}.
By Theorem \ref{thm1.3}, there exists a subset $S\subset[d]$ of size $k$ such that $\text{\rm rank}(\mathbf{A}_S)=k$ and
\begin{equation}\label{proofthm1.1:eq1}
\Vert \mathbf{A}-\mathbf{A}_{S}\mathbf{A}_{S}^{\dagger}\mathbf{A}\Vert_2^2=\text{\rm maxroot}\ p_{S}(x)\leq \text{\rm maxroot}\ (x^2\cdot\partial_x-d\cdot x)^k \ \det[x\cdot \mathbf{I}_d-\mathbf{A}^{\rm T}\mathbf{A}].
\end{equation}	
By Theorem \ref{thm3.1}, for each integer $k$ satisfying $\beta\cdot t\leq k< t$,  we have
\begin{equation}\label{eq:proofeq2}
\text{\rm maxroot}\ (x^2\cdot\partial_x-d\cdot x)^k \ \det[x\cdot \mathbf{I}_d-\mathbf{A}^{\rm T}\mathbf{A}]\leq \frac{\lambda_1}{1+(\frac{\lambda_1}{\alpha}-1) \cdot\bigg(\sqrt{\frac{k}{t}}-\sqrt{\frac{\beta}{1-\beta}\cdot (1-\frac{k}{t})} \bigg)^2}.
\end{equation}
Combining (\ref{eq:proofeq2}) with \eqref{proofthm1.1:eq1} and noting that $\lambda_1=\Vert\mathbf{A}\Vert_2^2$, we arrive at our conclusion.

\end{proof}

For the remainder of this subsection, we present two useful formulas, i.e., (\ref{prop3.1:eq1}) and
(\ref{prop3.1:eq2}).  Specifically, (\ref{prop3.1:eq2}) sheds light on the significance of the constant $\alpha$, which was introduced in Theorem \ref{thm1.1}.
For convenience, given a matrix $\mathbf{A}\in\mathbb{R}^{n\times d}$ and a positive integer $k\leq\mathrm{rank}(\mathbf{A})$, we define $c_{k}(\mathbf{A})$ as
\begin{equation*}
c_{k}(\mathbf{A}):=\sum\limits_{S\subset[d],\vert S\vert=k}\det[\mathbf{A}_S^{\rm T}\mathbf{A}_S].
\end{equation*}
A simple observation is that $c_{k}(\mathbf{A})$ is positive for each $k\leq\mathrm{rank}(\mathbf{A})$.

\begin{prop}\label{prop3.1}
Let $\mathbf{A}\in\mathbb{R}^{n\times d}$ be a matrix of rank $t\leq\min\{n,d\}$. Then, for each $k\leq t$, we have
\begin{equation}\label{prop3.1:eq1}
k!\sum\limits_{S\subset[d],\vert S\vert=k}^{}\det[\mathbf{A}_S^{\rm T}\mathbf{A}_S]\cdot p_S(x)=(x^2\cdot\partial_x-d\cdot x)^k\  \det[x\cdot\mathbf{I}_d-\mathbf{A}^{\rm T}\mathbf{A}].
\end{equation}	
Moreover, by letting $\alpha:=t/\sum_{i=1}^t \lambda_i^{-1}$, where $\lambda_1,\ldots,\lambda_t$ are the positive eigenvalues of $\mathbf{A}^{\rm T}\mathbf{A}$, we have
\begin{equation}\label{prop3.1:eq2}
\alpha=  \sum\limits_{S\subset[d],\vert S\vert=t-1}^{} \frac{\det[\mathbf{A}_S^{\rm T}\mathbf{A}_S]}{c_{t-1}(\mathbf{A})}\cdot \Vert \mathbf{A}-\mathbf{A}_S\mathbf{A}_S^{\dagger}\mathbf{A} \Vert_2^2.
\end{equation}	
\end{prop}

\begin{proof}
We first prove \eqref{prop3.1:eq1} by induction on $k$. The case $k=1$ immediately follows from Lemma \ref{lemma3.1} by taking $S=\emptyset$. We assume that \eqref{prop3.1:eq1} holds for each integer $i\leq k-1$. Then, by induction, we have
\begin{equation*}
\begin{aligned}
(x^2\cdot\partial_x-d\cdot x)^k\  \det[x\cdot\mathbf{I}_d-\mathbf{A}^{\rm T}\mathbf{A}]
&=(x^2\cdot\partial_x-d\cdot x)\cdot (k-1)!\cdot \sum\limits_{S\subset[d],\vert S\vert=k-1}^{}\det[\mathbf{A}_S^{\rm T}\mathbf{A}_S]\cdot p_S(x)\\
&=(k-1)!\cdot \sum\limits_{S\subset[d],\vert S\vert=k-1}^{}\sum\limits_{i\notin S}\det[\mathbf{A}_S^{\rm T}\mathbf{A}_S]\cdot  \Vert \mathbf{Q}_S\mathbf{a}_i\Vert^2\cdot p_{S\cup\{i\}}(x),
\end{aligned}
\end{equation*}	
where the last equality follows from Lemma \ref{lemma3.1}.
For each $S\subset[d]$ and each $i\notin S$, applying Lemma \ref{lemma2.2} with $\mathbf{B}=\mathbf{A}_S$ and $\mathbf{C}=\mathbf{a}_i$, we obtain
$\det[\mathbf{A}_S^{\rm T}\mathbf{A}_S]\cdot  \Vert \mathbf{Q}_S\mathbf{a}_i\Vert^2=\det[\mathbf{A}_{S\cup\{i\}}^{\rm T}\mathbf{A}_{S\cup\{i\}}]$.
Hence, we arrive at
\begin{equation*}
\begin{aligned}
(x^2\cdot\partial_x-d\cdot x)^k\  \det[x\cdot\mathbf{I}_d-\mathbf{A}^{\rm T}\mathbf{A}]&=(k-1)!\cdot \sum\limits_{S\subset[d],\vert S\vert=k-1}^{}\sum\limits_{i\notin S}\det[\mathbf{A}_{S\cup\{i\}}^{\rm T}\mathbf{A}_{S\cup\{i\}}]\cdot p_{S\cup\{i\}}(x)\\
&=k! \sum\limits_{R\subset[d],\vert R\vert=k}^{}\det[\mathbf{A}_R^{\rm T}\mathbf{A}_R]\cdot p_R(x).
\end{aligned}
\end{equation*}	
This completes the proof of \eqref{prop3.1:eq1}.

It remains to prove \eqref{prop3.1:eq2}. We use \eqref{prop3.1:eq1} for the case $k=t-1$. If $S$ is a $(t-1)$-subset of $[d]$ such that $\text{\rm rank}(\mathbf{A}_S)<t-1$, then we have $\det[\mathbf{A}_S^{\rm T}\mathbf{A}_S]=0$, implying that $\det[\mathbf{A}_S^{\rm T}\mathbf{A}_S]\cdot p_S(x)\equiv 0$; otherwise, we have $\det[\mathbf{A}_S^{\rm T}\mathbf{A}_S]>0$, and the polynomial $p_S(x)$ has only one positive root, which is exactly $\Vert \mathbf{A}-\mathbf{A}_S\mathbf{A}_S^{\dagger}\mathbf{A} \Vert_2^2$. Therefore, the left-hand side of \eqref{prop3.1:eq1} can be expressed as follows:
\begin{equation}\label{prop3.1:eq3}
(t-1)!\sum\limits_{S\subset[d],\vert S\vert=t-1}^{}\det[\mathbf{A}_S^{\rm T}\mathbf{A}_S]\cdot p_S(x)=(t-1)!\cdot \sum\limits_{\substack{S\subset[d],\vert S\vert=t-1,\\ \text{\rm rank}(\mathbf{A}_S)=t-1 }}^{}\det[\mathbf{A}_S^{\rm T}\mathbf{A}_S]\cdot x^{d-1}\cdot \bigg(x-\Vert \mathbf{A}-\mathbf{A}_S\mathbf{A}_S^{\dagger}\mathbf{A} \Vert_2^2\bigg).
\end{equation}
Moreover, using \eqref{prop2.1:eq1}, we can rewrite the right-hand side of \eqref{prop3.1:eq1} as
\begin{equation}\label{prop3.1:eq4}
(x^2\cdot\partial_x-d\cdot x)^{t-1}\ \det[x\cdot\mathbf{I}_d-\mathbf{A}^{\rm T}\mathbf{A}]=t! \cdot \lambda_1\cdots\lambda_t\cdot x^{d-1} \cdot \frac{1}{\alpha}\cdot (x-\alpha).
\end{equation}
Comparing \eqref{prop3.1:eq1},  \eqref{prop3.1:eq3}, and \eqref{prop3.1:eq4}, we obtain \eqref{prop3.1:eq2}.

\end{proof}

\begin{remark}
According to \eqref{prop3.1:eq1}, the polynomial $\frac{1}{k!\cdot c_{k}(\mathbf{A})}\cdot(x^2\cdot\partial_x-d\cdot x)^k\  \det[x\cdot\mathbf{I}_d-\mathbf{A}^{\rm T}\mathbf{A}]$ is the expectation of the polynomials $p_S(x)$ over the volume sampling distribution on the $k$-subsets of $[d]$. Moreover, equation \eqref{prop3.1:eq1} also implies that $(k+1)\cdot c_{k+1}(\mathbf{A})/c_k(\mathbf{A})$ is the expectation of $\Vert \mathbf{A}-\mathbf{A}_S\mathbf{A}_S^{\dagger}\mathbf{A}\Vert_{\rm F}^2$ over the volume sampling distribution on the $k$-subsets of $[d]$, i.e.,
\begin{equation}\label{remark3.2:eq1}
\sum\limits_{S\subset[d],\vert S\vert=k}^{}\frac{\det[\mathbf{A}_S^{\rm T}\mathbf{A}_S]}{c_k(\mathbf{A})}\cdot \Vert \mathbf{A}-\mathbf{A}_S\mathbf{A}_S^{\dagger}\mathbf{A}\Vert_{\rm F}^2=\frac{(k+1)\cdot c_{k+1}(\mathbf{A})}{c_k(\mathbf{A})}.
\end{equation}
Specifically, using \eqref{prop2.1:eq1} for each $k\leq t$, we can rewrite the right-hand side of \eqref{prop3.1:eq1} as
%\begingroup\fontsize{9pt}{12pt}\selectfont
%\begin{equation*}
%%\label{remark3.2:eq0}
%(x^2\cdot\partial_x-d\cdot x)^k\  \det[x\cdot\mathbf{I}_d-\mathbf{A}^{\rm T}\mathbf{A}]= (-1)^k\cdot \mathcal{R}_{x,d}^+\cdot \partial_x^k\cdot \mathcal{R}_{x,d}^+\ \bigg(x^{d-t}\prod\limits_{i=1}^{t}(x-\lambda_i)\bigg)=\sum\limits_{i=k}^t\frac{(-1)^{i-k}\cdot i!}{(i-k)!}\cdot x^{d-i+k}\cdot \sum\limits_{S\subset[t],\vert S\vert=i}\prod\limits_{j\in S}\lambda_j.
%\end{equation*}
%\endgroup
%\begin{equation*}
%%\begin{aligned}
%%\label{remark3.2:eq0}
%(x^2\cdot\partial_x-d\cdot x)^k\  \det[x\cdot\mathbf{I}_d-\mathbf{A}^{\rm T}\mathbf{A}]= (-1)^k\cdot \mathcal{R}_{x,d}^+\cdot \partial_x^k\cdot \mathcal{R}_{x,d}^+\ \bigg(x^{d-t}\prod\limits_{i=1}^{t}(x-\lambda_i)\bigg)=\sum\limits_{i=k}^t p_{i,k}\cdot x^{d-i+k}\cdot \sum\limits_{S\subset[t],\vert S\vert=i}\prod\limits_{j\in S}\lambda_j,
%%\end{aligned}
%\end{equation*}
%where $p_{i,k}:=\frac{(-1)^{i-k}\cdot i!}{(i-k)!}$.
\begin{equation*}
\begin{aligned}
%\label{remark3.2:eq0}
(x^2\cdot\partial_x-d\cdot x)^k\  \det[x\cdot\mathbf{I}_d-\mathbf{A}^{\rm T}\mathbf{A}]&= (-1)^k\cdot \mathcal{R}_{x,d}^+\cdot \partial_x^k\cdot \mathcal{R}_{x,d}^+\ \bigg(x^{d-t}\prod\limits_{i=1}^{t}(x-\lambda_i)\bigg)\\
&=\sum\limits_{i=k}^t\frac{(-1)^{i-k}\cdot i!}{(i-k)!}\cdot x^{d-i+k}\cdot \sum\limits_{S\subset[t],\vert S\vert=i}\prod\limits_{j\in S}\lambda_j.
\end{aligned}
\end{equation*}
Then, comparing the coefficients of $x^{d}$ and $x^{d-1}$ on the both sides of \eqref{prop3.1:eq1}, we obtain that for each $k<t$,
\begin{equation*}
%\label{remark3.2:eq0}
c_k(\mathbf{A}) =\sum\limits_{S\subset[t],\vert S\vert=k}\prod\limits_{j\in S}\lambda_j
\end{equation*}
and
\begin{equation*}
%\label{remark3.2:eq0}
\sum\limits_{S\subset[d],\vert S\vert=k}\det[\mathbf{A}_S^{\rm T}\mathbf{A}_S]  \cdot \Vert \mathbf{A}-\mathbf{A}_S\mathbf{A}_S^{\dagger}\mathbf{A}\Vert_{\rm F}^2 =(k+1)\sum\limits_{S\subset[t],\vert S\vert=k+1}\prod\limits_{j\in S}\lambda_j=(k+1)\cdot c_{k+1} (\mathbf{A}),
\end{equation*}
which immediately imply \eqref{remark3.2:eq1}.
Here, we use the fact that $\Vert \mathbf{A}-\mathbf{A}_S\mathbf{A}_S^{\dagger}\mathbf{A}\Vert_{\rm F}^2$ equals $\mathrm{Tr}[\mathbf{A}^{\rm T}\mathbf{Q}_S\mathbf{A}]$, which is the absolute value of the coefficient of $x^{d-1}$ in $p_S(x)$. Equation \eqref{remark3.2:eq1} recovers the first equation found in {\rm \cite[Equation (3.4)]{DRVW06}}, which played an important role in their analysis of Problem \ref{problem1.1} in the Frobenius norm case.

\end{remark}

\subsection{A deterministic polynomial time algorithm and the proof of Theorem  \ref{thm3.2}}\label{section3.4}

For convenience, throughout this subsection, we use $\lambda_{\max}^{\varepsilon}(p)\in \mathbb{R}$ to denote an $\varepsilon$-approximation to the largest root of a real-rooted polynomial $p(x)$, i.e.,
\begin{equation}\label{eps-approx}
| \lambda_{\max}^{\varepsilon}(p)-\text{\rm maxroot}\ p|\leq \varepsilon	.
\end{equation}
 Recall that $\mathbf{A}\in\mathbb{R}^{n\times d}$ is a matrix of rank $t$ and let $k\in[\beta\cdot t,t)$ be the sampling parameter, where $\beta$ is defined in Theorem \ref{thm1.1}. To simplify our notation, we define the degree-$d$ polynomial $\widehat{p}_S(x)$ for each subset $S\subset[d]$ as follows:
\begin{equation}
\widehat{p}_S(x):=(x^2\cdot\partial_x-d\cdot x)^{k-|S|}\ p_S(x)=(x^2\cdot\partial_x-d\cdot x)^{k-|S|}\ \det[x\cdot\mathbf{I}_d-\mathbf{A}^{\rm T}\mathbf{Q}_S\mathbf{A}].	
\end{equation}

 Building upon the proof of  Theorem \ref{thm1.3}, we propose a greedy algorithm that  selects columns from the matrix $\mathbf{A}$ in an iterative manner. Suppose that we have identified a subset $S\subset[d]$ during the $(l-1)$-th iteration. Then, during the $l$-th iteration, our algorithm will identify an index $i\in[d]\backslash S$ that minimizes the largest root of $\widehat{p}_{S\cup {i}}(x)$. This can be achieved by computing an $\varepsilon$-approximation of the largest root of $\widehat{p}_{S\cup {i}}(x)$ for every $i\in[d]\backslash S$.
Our deterministic greedy selection algorithm is stated in Algorithm \ref{alg1}. To save on running time, we store the matrix $\mathbf{A}^{\rm T}\mathbf{Q}_{S}\mathbf{A}$ at each iteration
and use the rank-one update formula to compute $\mathbf{A}^{\rm T}\mathbf{Q}_{S\cup\{i\}}\mathbf{A}$, i.e.,
\begin{equation*}
\mathbf{A}^{\rm T}\mathbf{Q}_{S\cup\{i\}}\mathbf{A}=	\mathbf{A}^{\rm T}\left(\mathbf{Q}_{S}-\frac{\mathbf{Q}_S\mathbf{a}_i \mathbf{a}_i^{\rm T} \mathbf{Q}_S }{ \Vert \mathbf{Q}_S\mathbf{a}_i\Vert^2}\right)\mathbf{A}=\mathbf{A}^{\rm T}\mathbf{Q}_{S}\mathbf{A}-\frac{(\mathbf{A}^{\rm T}\mathbf{Q}_S\mathbf{a}_i)(\mathbf{A}^{\rm T}\mathbf{Q}_S\mathbf{a}_i)^{\rm T} }{ \Vert \mathbf{Q}_S\mathbf{a}_i\Vert^2}.
\end{equation*}
Furthermore, to reduce the time complexity of computing $p_{S\cup\{i\}}(x)$, we work with the $n\times n$ matrix $\mathbf{Q}_{S\cup\{i\}}\mathbf{A}\mathbf{A}^{\rm T}$ when $d\geq n$ and compute $p_{S\cup\{i\}}(x)$ using the following formula:
\begin{equation*}
p_{S\cup\{i\}}(x)=\det[x\cdot \mathbf{I}_d-\mathbf{A}^{\rm T}\mathbf{Q}_{S\cup\{i\}}\mathbf{A}]=x^{d-n}\cdot \det[x\cdot \mathbf{I}_n-\mathbf{Q}_{S\cup\{i\}}\mathbf{A}\mathbf{A}^{\rm T}].
\end{equation*}
In the following we present a proof of Theorem \ref{thm3.2}, showing that Algorithm \ref{alg1} can attain the bound in Theorem \ref{thm1.1} up to a certain computational error.

\begin{algorithm}[t]
%[!h]
\caption{A deterministic greedy algorithm for column subset selection}\label{alg1}
\begin{algorithmic}[1]
\Require $\mathbf{A}=[\mathbf{a}_1,\ldots,\mathbf{a}_d]\in\mathbb{R}^{n\times d}$ of rank $t$; sampling parameter $k\in[\beta\cdot t,t)$; $\varepsilon>0$.
\Ensure A subset $S\subset[d]$ of cardinality $k$ satisfying the bound in \eqref{alg:eq1}.
\State Set $S=\emptyset$. Set $\mathbf{Q}=\mathbf{I}_n$ and $\mathbf{B}=\begin{cases}\mathbf{A}^{\rm T}\mathbf{Q}\mathbf{A}& \text{if $d<n$} \\
\mathbf{Q}\mathbf{A}\mathbf{A}^{\rm T} & \text{if $d\geq n$}	
\end{cases}$
\For{$l=1,2,\ldots,k$}

\State Compute $\mathbf{Q}\mathbf{a}_i\in\mathbb{R}^{n\times 1}$ and $\mathbf{Q}_{S\cup\{i\}}=\mathbf{Q}-\frac{\mathbf{Q}\mathbf{a}_i\cdot (\mathbf{Q}\mathbf{a}_i)^{\rm T} }{\Vert \mathbf{Q}\mathbf{a}_i\Vert^2}\in\mathbb{R}^{n\times n}$ for each $i\in[d]\backslash S$.
\If{$d< n$}
\State Compute $\mathbf{A}^{\rm T}\mathbf{Q}\mathbf{a}_i\in\mathbb{R}^{d\times 1}$ and  $\mathbf{A}^{\rm T}\mathbf{Q}_{S\cup\{i\}}\mathbf{A}=\mathbf{B}-\frac{\mathbf{A}^{\rm T}  \mathbf{Q}\mathbf{a}_i\cdot (\mathbf{A}^{\rm T}  \mathbf{Q}\mathbf{a}_i)^{\rm T} }{\Vert \mathbf{Q}\mathbf{a}_i\Vert^2} \in\mathbb{R}^{d\times d}$ for each $i\in[d]\backslash S$.
\State Compute the polynomial $p_{S\cup\{i\}}(x)=\det[x\cdot \mathbf{I}_d-\mathbf{A}^{\rm T}\mathbf{Q}_{S\cup\{i\}}\mathbf{A}]$ for each $i\in[d]\backslash S$.
\Else
\State Compute $\mathbf{a}_i^{\rm T}  \mathbf{B}\in\mathbb{R}^{1\times n}$ and  $\mathbf{Q}_{S\cup\{i\}}\mathbf{A}\mathbf{A}^{\rm T}=\mathbf{B}-\frac{\mathbf{Q}\mathbf{a}_i\cdot \mathbf{a}_i^{\rm T}  \mathbf{B} }{\Vert \mathbf{Q}\mathbf{a}_i\Vert^2} \in\mathbb{R}^{n\times n}$ for each $i\in[d]\backslash S$.
\State Compute the characteristic polynomial $\det[x\cdot \mathbf{I}_n-\mathbf{Q}_{S\cup\{i\}}\mathbf{A}\mathbf{A}^{\rm T}]$ for each $i\in[d]\backslash S$.
\State Compute the polynomial $p_{S\cup\{i\}}(x)=\det[x\cdot \mathbf{I}_d-\mathbf{A}^{\rm T}\mathbf{Q}_{S\cup\{i\}}\mathbf{A}]$ for each $i\in[d]\backslash S$.
\EndIf
\State Compute the polynomial $\widehat{p}_{S\cup\{i\}}(x)=(x^2\cdot \partial_x-d\cdot x)^{k-l}\ p_{S\cup\{i\}}(x)$ for each $i\in[d]\backslash S$.
\State Using the standard technique of binary search with a Sturm sequence, compute an $\varepsilon$-approximation $\lambda_{\max}^{\varepsilon}(\widehat{p}_{S\cup\{i\}})$ to the largest root of $\widehat{p}_{S\cup\{i\}}(x)$ for each $i\in[d]\backslash S$.
\State Find
\begin{equation*}
j_l=\mathop{\mathrm{argmin}}_{i\in [d]\backslash S} 	\lambda_{\max}^{\varepsilon}(\widehat{p}_{S\cup\{i\}}).
\end{equation*}
\label{line 14}
\State Set $\mathbf{Q}=\mathbf{Q}_{S\cup\{j_l\}}$ and $\mathbf{B}=\begin{cases}\mathbf{A}^{\rm T}\mathbf{Q}_{S\cup\{j_l\}}\mathbf{A}& \text{if $d<n$} \\
\mathbf{Q}_{S\cup\{j_l\}}\mathbf{A}\mathbf{A}^{\rm T} & \text{if $d\geq n$}	
\end{cases}$.
\State  $S\gets S\cup \{j_l\}$.
\EndFor \\
\Return $S=\{j_1,j_2,\ldots,j_k\}$.
\end{algorithmic}
\end{algorithm}

\begin{proof}[{{Proof of Theorem \ref{thm3.2}}}]
We will first prove \eqref{alg:eq1}. At the $l$-th iteration, we let 
\begin{equation*}
i_l:=\mathop{\mathrm{argmin}}_{i:\Vert \mathbf{Q}\mathbf{a}_i \Vert\neq 0} 	\text{\rm maxroot}\  \widehat{p}_{S\cup\{i\}}(x).
\end{equation*}
Then, by our choice of $j_l$ in Line 14 of Algorithm \ref{alg1}, we have
\begin{equation}\label{thm3.2:eq1}
\text{\rm maxroot}\  \widehat{p}_{S\cup\{j_l\}}\overset{(a)}  \leq \lambda_{\max}^{\varepsilon}(\widehat{p}_{S\cup\{j_l\}})+\varepsilon \leq \lambda_{\max}^{\varepsilon}(\widehat{p}_{S\cup\{i_l\}})+\varepsilon \overset{(b)} \leq  \text{\rm maxroot}\  \widehat{p}_{S\cup\{i_l\}}+ 2\varepsilon \overset{(c)}\leq  \text{\rm maxroot}\ \widehat{p}_{S}  +2\varepsilon,
\end{equation}
where (a) and (b) follow from \eqref{eps-approx} and (c) follows from Lemma \ref{lemma3.2}.
By repeatedly using \eqref{thm3.2:eq1} for $l=1,2,\ldots,k$, we can obtain that
\begin{equation}\label{eq:thm3.2ineq}
\text{\rm maxroot}\ p_{\{j_1,\ldots,j_k\}}(x)\leq 2k\varepsilon+\text{\rm maxroot}\ (x^2\cdot\partial_x-d\cdot x)^k \ \det[x\cdot \mathbf{I}_d-\mathbf{A}^{\rm T}\mathbf{A}].
\end{equation}
Combining \eqref{eq:thm3.2ineq}, \eqref{section3.1:eq2} and Theorem \ref{thm3.1}, we arrive at \eqref{alg:eq1}.

Next, we will analyze the time complexity of Algorithm \ref{alg1} and divide the proof into two cases.
\begin{itemize}
\item {\em Case 1: $d<n$.}
 We need to compute the matrix $\mathbf{B}=\mathbf{A}^{\rm T}\mathbf{A}$ in Line 1 of Algorithm \ref{alg1} with a time complexity of $O(nd^2)$. We claim that the time complexity for computing all $\lambda_{\max}^{\varepsilon}(\widehat{p}_{S\cup\{i\}}), i\in [d]\setminus S,$ at each iteration is $O(dn^2+d^{w+1}\log(d\vee\varepsilon^{-1}))$. Then, since $d<n$, the total running time of Algorithm \ref{alg1} is
 \[
 O(nd^2)+O(k(dn^2+d^{w+1}\log(d\vee\varepsilon^{-1})))=O(kdn^2+kd^{w+1}\log(d\vee\varepsilon^{-1})).
 \]

It remains to prove that the time complexity for each iteration is $O(dn^2+d^{w+1}\log(d\vee\varepsilon^{-1}))$.
At the $l$-th iteration, given an integer $i \in [d]\setminus S$, we require $O(n^2)$ time to calculate $\mathbf{Q}\mathbf{a}_i$ and $\mathbf{Q}_{S\cup{i}}$ in Line 3 of Algorithm \ref{alg1}.
 After obtaining $\mathbf{Q}\mathbf{a}_i$, we can calculate  $\mathbf{A}^{\rm T}\mathbf{Q}\mathbf{a}_i$ in  $O(nd)$ time and $\mathbf{A}^{\rm T}\mathbf{Q}_{S\cup\{i\}}\mathbf{A}$ in  $O(d^2)$ time. Then, it takes time $O(d^w\log d)$ to compute the characteristic polynomial of $\mathbf{A}^{\rm T}\mathbf{Q}_{S\cup\{i\}}\mathbf{A}$ in Line 6 of Algorithm \ref{alg1} \cite{KG85,BCS97}. By employing  \eqref{prop2.1:eq1}, the time complexity for calculating the  the coefficients of the polynomial $\widehat{p}_{S\cup\{i\}}(x)$ is $O(1)$. Note that $\widehat{p}_{S\cup\{i\}}(x)$ has at most $d$ nonzero roots, so the time complexity to  obtain an $\varepsilon$-approximation to its largest root is $O(d^2\log{\varepsilon^{-1}} )$   (see \cite[Section 4.1]{inter3}). Recall that $d<n$, so, for each $i\in[d]\backslash S$,  the time complexity to obtain $\lambda_{\max}^{\varepsilon}(\widehat{p}_{S\cup\{i\}})$ is
\[
O(n^2)+O(nd)+O(d^2)+O(d^w\log d)+O(1)+O(d^2\log {\varepsilon^{-1}})=O(n^2+d^w\log(d\vee\varepsilon^{-1})).
\]
 Therefore, the total running time at each iteration is $O(dn^2+d^{w+1}\log(d\vee\varepsilon^{-1}))$.

\item
{\em Case 2: $d\geq n$.} In this case the time complexity for computing the matrix $\mathbf{B}=\mathbf{A}\mathbf{A}^{\rm T}$ in Line 1 is $O(dn^2)$. We claim that the time complexity for computing all $\lambda_{\max}^{\varepsilon}(\widehat{p}_{S\cup\{i\}}), i\in [d]\setminus S,$ at each iteration is $O(dn^w\log(n\vee\varepsilon^{-1}))$. Then, since $d\geq n$, the total running time for Algorithm \ref{alg1} is
\[
O(dn^2)+O(kdn^w\log(n\vee\varepsilon^{-1}))=O(kdn^w\log(n\vee\varepsilon^{-1})).
\]

It remains to show the time complexity for each iteration is $O(dn^w\log(n\vee\varepsilon^{-1}))$. First, we need  $O(n^2)$ time to compute the matrices $\mathbf{Q}\mathbf{a}_i$, $\mathbf{Q}_{S\cup\{i\}}$, $\mathbf{a}_i^{\rm T}  \mathbf{B}$ and $\mathbf{Q}_{S\cup\{i\}}\mathbf{A}\mathbf{A}^{\rm T}$ in Lines 3 and 8 of Algorithm \ref{alg1}. Next, we need  $O(n^w\log n)$ time to compute the characteristic polynomial $\det[x\cdot\mathbf{I}_n-\mathbf{Q}_{S\cup\{i\}}\mathbf{A}\mathbf{A}^{\rm T}]$ in Line 9. Then it takes time $O(1)$ to compute the polynomials $p_{S\cup\{i\}}(x)$ and $\widehat{p}_{S\cup\{i\}}(x)$. In the case of $d\geq n$, the polynomial $\widehat{p}_{S\cup\{i\}}(x)$ has at most $n$ nonzero roots. The time complexity to obtain an $\varepsilon$-approximation to its largest root is $O(n^2\log \varepsilon^{-1})$. Putting all of these together, the time complexity to obtain $\lambda_{\max}^{\varepsilon}(\widehat{p}_{S\cup\{i\}})$ for each $i\in[d]\backslash S$ is
 $O(n^2)+O(n^w\log n)+O(1)+O(n^2\log \frac{1}{\varepsilon})=O(n^w\log(n\vee\varepsilon^{-1}))$.  So the total running time at each iteration is $O(dn^w\log(n\vee\varepsilon^{-1}))$.
\end{itemize}

\end{proof}

%\section{Funding}
%{The work of Jian-Feng Cai was supported in part by the Hong Kong Research Grants Council GRF under Grants 16310620 and 16306821, and in part by Hong Kong Innovation and Technology Fund MHP/009/20. The work of Zhiqiang Xu is supported by the National Science Fund for Distinguished Young Scholars (12025108) and the National Nature Science Foundation of China  (12021001, 12288201).}

%\bibliographystyle{abbrv}
%\bibliography{ref}

\begin{thebibliography}{0}

\bibitem{And1}
Anderson, J.
``Extensions, restrictions, and representations of states on $C^*$-algebras.''
\textit{Trans. Amer. Math. Soc.} 249, no. 2 (1979): 303-329.

\bibitem{And2}
Anderson, J.
``Extreme points in sets of positive linear maps on $\mathcal{B}(\mathcal{H})$.''
\textit{J. Funct. Anal.} 31, no. 2 (1979): 195-217.

\bibitem{And3}
Anderson, J.
``A conjecture concerning the pure states of $\mathcal{B}(\mathcal{H})$ and a related theorem.''
In \textit{Topics in Modern Operator Theory: 5th International Conference on Operator Theory, Timi\c{s}oara and Herculane (Romania),} June 2–12, 1980, pp. 27-43. Basel: Birkh${\rm \ddot{a}}$user Basel, 1981.

\bibitem{ADMMWG15}
Anderson, D., S. Du, M. Mahoney, C. Melgaard, K. Wu, and M. Gu.
``Spectral gap error bounds for improving CUR matrix decomposition and the ${Nystr\ddot{o}m}$ method.''
In \textit{Artificial Intelligence and Statistics,} (2015), pp. 19-27. PMLR.

\bibitem{BRN10}
Balzano, L., B. Recht, and R. Nowak.
``High-dimensional matched subspace detection when data are missing.''
In \textit{2010 IEEE International Symposium on Information Theory,} (2010) pp. 1638-1642.

\bibitem{inter0}
Batson, J., D. A. Spielman, and N. Srivastava.
``Twice-Ramanujan sparsifiers.''
\textit{SIAM J. Comput.} 41, no. 6 (2012): 1704-1721. 

\bibitem{BBC20}
Belhadji, A., R. Bardenet, and P. Chainais.
``A determinantal point process for column subset selection.''
\textit{J. Mach. Learn. Res.} 21, no. 1 (2020): 8083-8144.

\bibitem{BDM14}
Boutsidis, C., P. Drineas, and M. Magdon-Ismail.
``Near-optimal column-based matrix reconstruction.''
\textit{SIAM J. Comput.} 43, no. 2 (2014): 687-717.

\bibitem{BMD09}
Boutsidis, C., M. W. Mahoney, and P. Drineas.
``An improved approximation algorithm for the column subset selection problem.''
In \textit{Proceedings of the twentieth annual ACM-SIAM symposium on Discrete algorithms,} (2019) pp. 968-977. Society for Industrial and Applied Mathematics.

\bibitem{BT89}
Bourgain, J. and L. Tzafriri.
``Restricted invertibility of matrices and applications.''
In \textit{Analysis at Urbana,} vol. II, pp. 61-107. 1989.

\bibitem{bcms19}
Bownik, M., P. Casazza, A. Marcus, and D. Speelge.
``Improved bounds in Weaver and Feichtinger conjectures.''
\textit{J. Reine Angew. Math.} 2019, no. 749 (2019): 267-293.

\bibitem{Branden2}
${\rm Br \ddot{a}nd \acute{e}n}$, P.
``Hyperbolic polynomials and the Kadison-Singer problem.''
(2018): preprint arXiv:1809.03255.


\bibitem{BCS97}
${\rm B\ddot{u}rgisser}$, P., M. Clausen, and M. A. Shokrollahi.
Algebraic complexity theory.
Vol. 315. Springer Science \& Business Media, 2013.


\bibitem{coh2016}
Cohen, M.
``Improved spectral sparsification and Kadison-Singer for sums of higher-rank matrices.'' 
Banff International Research Station for Mathematical Innovation and Discovery (2016).

\bibitem{CT}
Casazza P. G. and J. C. Tremain.
``The Kadison-Singer problem in mathematics and engineering.''
\textit{Proc. Natl. Acad. Sci. USA,} 103, no. 7 (2006): 2032-2039.

\bibitem{CH92}
Chan, T. F. and P. C. Hansen.
``Some applications of the rank revealing QR factorization.''
\textit{SIAM Journal on Scientific and Statistical Computing,} 13, no. 3 (1992): 727-741.


\bibitem{DKM20}
Derezinski, M., R. Khanna, and M. W. Mahoney.
``Improved guarantees and a multiple-descent curve for column subset selection and the ${ Nystr\ddot{o}m}$ method.''
\textit{Advances in Neural Information Processing Systems,} 33 (2020): 4953-4964.


\bibitem{DR10}
Deshpande, A. and L. Rademacher.
``Efficient volume sampling for row/column subset selection.''
In \textit{2010 IEEE 51st annual symposium on foundations of computer science,} pp 329-338. IEEE, 2010.

\bibitem{DRVW06}
Deshpande, A., L. Rademacher, S. Vempala, and G. Wang.
``Matrix approximation and projective clustering via volume sampling.''
\textit{Theory Comput.} 2, no. 1 (2006): 225-247.

\bibitem{DMM08}
Drineas, P., M. W. Mahoney, and S. Muthukrishnan.
``Relative-error CUR matrix decompositions.''
\textit{SIAM J. Matrix Anal. Appl.} 30, no. 2 (2008): 844-881.

\bibitem{Fell}
Fell, H. J. 
``On the zeros of convex combinations of polynomials.''
\textit{Pacific J. Math.}  89, no. 1 (1980):43-50.

\bibitem{You3}
Friedland, O. and P. Youssef.
``Approximating Matrices and Convex Bodies.''
\textit{Int. Math. Res. Not. IMRN} 2019, no. 8 (2019): 2519-2537.

\bibitem{FKV04}
Frieze, A., R. Kannan, and S. Vempala.
``Fast Monte-Carlo algorithms for finding low-rank approximations.''
\textit{J. ACM} 51, no. 6 (2004): 1025-1041.

\bibitem{Gol65}
Golub, G. H.
``Numerical methods for solving linear least squares problems.''
\textit{Numer. Math.} 7 (1965): 206-216.

\bibitem{GE96}
Gu M. and S. C. Eisenstat.
``Efficient algorithms for computing a strong rank-revealing QR factorization.''
\textit{SIAM J. Sci. Comput.} 17, no. 4 (1996): 848-869.

\bibitem{GE03}
Guyon I. and A. Elisseeff.
``An introduction to variable and feature selection.''
\textit{J. Mach. Learn. Res.} 3, no. Mar (2003): 1157-1182.

\bibitem{HP92}
Hong Y. P. and C. T. Pan.
``Rank-revealing QR factorizations and the singular value decomposition.''
\textit{Math. Comp.} 58, no. 197 (1992): 213-232.

\bibitem{KS59}
Kadison R. V. and I. M. Singer.
``Extensions of pure states.''
\textit{Amer. J. Math.} 81, no. 2 (1959): 383-400.

\bibitem{KG85}
Keller-Gehrig, W.
``Fast algorithms for the characteristics polynomial.''
\textit{Theoret. Comput. Sci.} 36 (1985): 309-317.

\bibitem{inter1}
Marcus, A. W., D. A. Spielman, and N. Srivastava.
``Interlacing families I: Bipartite Ramanujan graphs of all degrees.''
\textit{Ann. of Math. (2)} 182, (2015): 307-325.

\bibitem{inter2}
Marcus, A. W., D. A. Spielman, and N. Srivastava.
``Interlacing families II: mixed characteristic polynomials and the Kadison-Singer problem.''
\textit{Ann. of Math. (2)} 182, (2015): 327-350.

\bibitem{inter3}
Marcus, A. W., D. A. Spielman, and N. Srivastava.
``Interlacing Families III: Sharper restricted invertibility estimates.''
\textit{Israel J. Math.} 247, (2022), 519-546.

\bibitem{inter5}
Marcus, A. W., D. A. Spielman, and N. Srivastava.
``Finite free convolutions of polynomials.''
\textit{Probab. Theory Related Fields} 182, no. 3-4 (2022): 807-848.

\bibitem{Mar66}
Marden, M.
Geometry of polynomials.
No. 3. American Mathematical Soc., 1949.

\bibitem{MP07}
McGown, K. J. and H. R. Parks.
``The generalization of Faulhaber's formula to sums of non-integral powers.''
\textit{J. Math. Anal. Appl.} 330, no. 1 (2007): 571-575.


\bibitem{You2}
Naor, A. and P. Youssef.
``Restricted invertibility revisited.''
In \textit{A Journey Through Discrete Mathematics,} 657-691. Springer, 2017

\bibitem{ravi1}
Ravichandran, M.
``Principal submatrices, restricted invertibility and a quantitative Gauss-Lucas theorem.''
\textit{Int. Math. Res. Not. IMRN} 2020, no. 15 (2020): 4809-4832.

\bibitem{ravi2}
Ravichandran, M. and J. Leake.
``Mixed determinants and the Kadison-Singer problem.''
\textit{Math. Ann.} 377, no. 1-2 (2020): 511-541.

\bibitem{ravi3}
Ravichandran, M. and N. Srivastava.
``Asymptotically Optimal Multi-Paving.''
\textit{Int. Math. Res. Not. IMRN} 2021, no. 14 (2021): 10908-10940.

\bibitem{RV07}
Rudelson, M. and R. Vershynin.
``Sampling from large matrices: An approach through geometric functional analysis.''
\textit{J. ACM} 54, no. 4 (2007): 21-es.

\bibitem{SS12}
Spielman, D. A. and N. Srivastava.
``An elementary proof of the restricted invertibility theorem.''
\textit{Israel J. Math.} 190, no. 1 (2012): 83-91.


\bibitem{Ver01}
Vershynin, R. 
``John's decompositions: selecting a large part.''
\textit{Israel J. Math.} 122 (2001): 253-277.


\bibitem{WS18}
Wang, Y. and A. Singh.
``Provably Correct Algorithms for Matrix Column Subset Selection with Selectively Sampled Data.''
\textit{J. Mach. Learn. Res.}  18, no. 1 (2017): 5699-5740.

\bibitem{Weaver}
Weaver, N.
``The Kadison-Singer problem in discrepancy theory.''
\textit{Discrete Math.} 278, no. 1-3 (2004): 227-239.

\bibitem{XXZ21}
Xu, Z., Z. Xu, and Z. Zhu.
``Improved bounds in Weaver's $\rm{KS}_r$ conjecture for high rank positive semidefinite matrices.''
\textit{J. Funct. Anal.} 285, no. 4 (2023): 109978

\bibitem{YTT11}
Yanai, H., K. Takeuchi, and Y. Takane.
Projection Matrices, Generalized Inverse Matrices, and Singular Value Decomposition.
Statistics for Social and Behavioral Sciences, Springer New York, 2011.

\bibitem{You14}
Youssef, P. 
``Restricted invertibility and the Banach-Mazur distance to the cube.''
\textit{Mathematika,} 60, no. 1 (2014): 201-218.

\end{thebibliography}

\Addresses

\end{document}